\documentclass[lettersize,journal]{IEEEtran}
  \usepackage{cite}
\usepackage{graphicx}
\usepackage{subfigure}
\usepackage{epstopdf}
\usepackage{multicol}
\usepackage{algorithm}
\usepackage{algorithmic}
\usepackage{amssymb}
\usepackage{amsthm,amsmath,amsfonts}
\usepackage{mathrsfs}
\usepackage{courier}
\newtheorem{lemma}{\textbf{Lemma}}

\usepackage{bm}
\usepackage{subfigure}
\usepackage{color}
\usepackage{ragged2e} 
\usepackage{ragged2e}
\usepackage{booktabs}

\ifCLASSINFOpdf
\else
\fi
\begin{document}
%
\title{Joint Source-Channel Optimization for UAV Video Coding and Transmission}
%
%
%
%

\author{Kesong Wu, Xianbin Cao,~\IEEEmembership{Senior Member,~IEEE}, Peng Yang,~\IEEEmembership{Member,~IEEE}, \\ Haijun Zhang,~\IEEEmembership{Fellow,~IEEE}, Tony Q. S. Quek,~\IEEEmembership{Fellow,~IEEE}, and 
Dapeng Oliver Wu,~\IEEEmembership{Fellow,~IEEE}

\thanks{
K. Wu, X. Cao, and P. Yang are with the School of Electronic and Information Engineering, Beihang University, Beijing 100191, China. (email: {sdtwks@163.com,\{xbcao,peng\_yang\}@buaa.edu.cn})

H. Zhang is with the Beijing Engineering and Technology Research Center for Convergence Networks and Ubiquitous Services, University of Science and Technology Beijing, Beijing 100083, China. (email: haijunzhang@ieee.org) 

T. Q. S. Quek is with the Information Systems Technology and Design Pillar, Singapore University of Technology and Design, Singapore 487372, Singapore. (email: tonyquek@sutd.edu.sg)

D. O. Wu is with the Department of Computer Science, City University of Hong Kong, Kowloon, Hong Kong, China. (email: dpwu@ieee.org)
}
}

\IEEEtitleabstractindextext{%
\justifying\let\raggedright\justifying
\begin{abstract}
This paper is concerned with unmanned aerial
vehicle (UAV) video coding and transmission in scenarios such as aerial search and monitoring. Unlike existing methods of modeling UAV video source coding and channel transmission separately, we investigate the joint source-channel optimization issue for video coding and transmission. Particularly, we design eight-dimensional delay-power-rate-distortion models in terms of source coding and channel transmission and characterize the correlation between video coding and transmission, with which a joint source-channel optimization problem is formulated. Its objective is to minimize end-to-end distortion and UAV power consumption by optimizing fine-grained parameters related to UAV video coding and transmission. This problem is confirmed to be a challenging sequential-decision and non-convex optimization problem. We therefore decompose it into a family of repeated optimization problems by Lyapunov optimization and design an approximate convex optimization scheme with provable performance guarantees to tackle these problems. Based on the theoretical transformation, we propose a Lyapunov repeated iteration (LyaRI) algorithm. Both objective and subjective experiments are conducted to comprehensively evaluate the performance of LyaRI. The results indicate that, compared to its counterparts, LyaRI achieves better video quality and stability performance, with a 47.74\% reduction in the variance of the obtained encoding bitrate.
\end{abstract}

}

\maketitle

\IEEEdisplaynontitleabstractindextext

%
\IEEEpeerreviewmaketitle

%
%
%
%
\section{Introduction}
\IEEEPARstart{B}{enefiting} from their distinctive advantages of low cost and high mobility, unmanned aerial vehicles (UAVs) have been increasingly adopted in more and more practical scenarios, 
such as geological disaster monitoring, forest fire detection, emergency communications, and emergency search and rescue \cite{DBLP:journals/tnse/WuCJWXXZXCLWW24,DBLP:journals/tnse/CuiLZZWBH24,Zhao2019UAVAssistedEN}.
In the application of UAVs, video acquisition and transmission constitute one of their pivotal missions. UAVs are capable of capturing targets through on-board cameras and subsequently transmitting coded and compressed videos via on-board data transmission modules. Rapid acquisition of dynamic information on site through UAV videos can save valuable time for emergency decision making. As a result, UAV video coding and transmission technology has emerged as a focal point of interest within both academia and industry \cite{DBLP:conf/infocom/ZhanHMW22,DBLP:journals/tase/HuangSN22}. 

However, the timely and efficient transmission of UAV video streams is fraught with considerable challenges \cite{DBLP:journals/tccn/BharGA24}. Firstly, a UAV network is bandwidth-constrained while the video is characterized by its substantial information and stringent timeliness. It is essential to mitigate information redundancy through video coding for UAV video transmission. Nevertheless, the computational complexity and energy expenditure associated with video coding are exceedingly high, presenting a significant challenge for UAVs that are subject to severe energy constraints. Secondly, the occurrence of congestion and packet loss, potentially leading to visual artifacts or stalling in video decoding at the receiver, may arise if the video coding bitrate surpasses the UAV channel capacity \cite{DBLP:journals/tcsv/LiuJ22}. 
Further, UAV channel capacity is subject to dynamic fluctuations, and ensuring the match of UAV video coding bitrate with UAV channel capacity is a formidable challenge.

\subsection{State of the Arts}
Recent years have witnessed a proliferation of research in the domain of UAV video transmission. The authors in \cite{DBLP:journals/tcsv/ZhanHSLWW21} proposed an approach to maximize video utility through the joint optimization of resource allocation and UAV trajectory, thereby reducing operational time of {the} UAV. In \cite{DBLP:journals/tmm/TangHH21}, the authors introduced a quality of experience (QoE) driven dynamic wireless video broadcasting scheme, aiming to maximize the peak signal-to-noise ratio (PSNR) of reconstructed video by jointly optimizing UAV transmit power and trajectory. The aforementioned studies focused on enhancing throughput of UAV networks and reducing transmission distortion by optimizing UAV trajectories, resource allocation, and transmit power without considering video coding distortion. To tackle this issue, the authors in \cite{DBLP:journals/icl/TangHSHY23} investigated the optimal deployment of UAVs in UAV-assisted virtual reality (VR) systems with the objective of minimizing video transmission delay while meeting the image resolution requirements. This research focused on resolution, an important indicator of video coding, but did not delve into video coding. 

Video coding is an essential technology for achieving video compression, and rate distortion optimization (RDO) is the core theoretical foundation of video coding. The objective of RDO is to minimize encoding distortion under a fixed coding bitrate or minimize the coding bitrate at a specific distortion level. Researchers have conducted extensive explorations in this field and have made remarkable progress. Under the constraint of coding complexity, the authors in  \cite{DBLP:journals/spic/LiYWK21} investigated how to optimize the global rate-distortion (R-D) between consecutive video frames to enhance the efficiency of video coding. By constructing an R-D optimization model based on a quantization step cascade method, the authors in \cite{DBLP:journals/tii/GaoJWMLK22} derived the quantization step of each frame to achieve optimal quantization parameter selection during the coding process. Besides, based on the establishment of an inter-frame R-D dependency model, the authors in  \cite{DBLP:journals/tip/GongYLLLW21} constructed the R-D cost function with the Lagrange optimization theory. They solved the theoretical model of the optimal quantization parameters of each frame. As a result, the coding bitrate was reduced on the premise of meeting the distortion requirements. 

All these studies assumed unlimited encoding time and power consumption. They also focused on the selection of video coding quantization parameters based on the R-D relationship and aimed to achieve the best encoding performance. However, the time and power consumption of video coding have been confirmed to significantly impact the video bitrate and distortion. In \cite{DBLP:journals/tce/SchimpfLL23}, the authors explored semi-extreme sparse coding of transformed residuals at medium and low bitrates to enhance video compression efficiency. The work in \cite{DBLP:journals/tcsv/LiWX14} indicated that video bitrate and distortion were primarily determined by the distribution characteristics of the transformed residuals. Transformed residuals were mainly influenced by motion estimation (ME) and quantization parameters. The accuracy of ME was closely related to the sum of absolute differences (SAD) of macroblocks. The number of calculations of macroblock SAD was proportional to the encoding complexity. Considering that the time and power consumption of video coding were increasing functions of encoding complexity, the authors in \cite{DBLP:journals/tcsv/LiWX14} deduced a quantitative relationship between video coding delay, power, rate, and distortion, and extended the traditional R-D model to a delay-power-rate-distortion (d-P-R-D) model. Based on this model, a new rate control (RC) method was proposed, aiming to minimize encoding distortion under constraints of rate, delay, and power. However, there are two aspects for improvement in this research. Firstly, the work is conducted for the earlier video coding standard advanced video coding (AVC). High efficiency video coding (HEVC) and versatile video coding (VVC), as the subsequent video coding standards\cite{DBLP:journals/tcsv/WeiZWYCK24}, offer lower coding bitrates under the same video quality conditions. Therefore, it is necessary to further investigate the d-P-R-D model for more advanced standards. Secondly, although the end-to-end delay composition of video transmission systems was analyzed in \cite{DBLP:journals/tcsv/LiWX14}, the proposed d-P-R-D model was only for video source coding. To achieve better end-to-end video transmission QoE, it is essential to design a joint source-channel video coding and transmission scheme. 

From the perspective of joint source-channel design, efficient video source coding can reduce the video bitrate while maintaining the same video quality. The reduction in video bitrate leads to decreasing transmit power consumption and transmission delay. However, efficient video source coding implies higher computational complexity, which result in increased power consumption and delay in video coding. Given the constraints of end-to-end latency, increasing encoding latency to achieve better compression performance may reduce the transmission latency budget, subsequently increasing the risk of distortion during transmission. Similarly, changes in encoding power consumption can affect the power consumption of video transmission under a given power constraint of video transmission system. Researchers have extensively {studied} joint source-channel optimization 
\cite{DBLP:conf/icassp/YangK22,DBLP:conf/globecom/ChenCYHWZ23,DBLP:journals/tmc/CenGM22,DBLP:journals/tcom/HuC21}. For instance, the work in \cite{DBLP:conf/icassp/YangK22} proposed a new scheme for dynamically adjusting the bitrate based on channel signal-to-noise ratio (SNR) and image content. In \cite{DBLP:conf/globecom/ChenCYHWZ23}, the authors presented a joint source-channel adaptive RC scheme for wireless image transmission based on image feature mapping, entropy awareness, and channel conditions. The authors in \cite{DBLP:journals/tmc/CenGM22} introduced a novel video codec structure based on compressed sensing. They also investigated the minimization of system power consumption by jointly optimizing video coding bitrate and multi-path transmission bitrate under distortion and energy consumption constraints. The work in \cite{DBLP:journals/tcom/HuC21} investigated the trade-off between latency, power, and distortion and achieved deep cross-layer optimization of video coding, queuing control, and wireless transmission through a joint lossy compression and power allocation scheme. However, the aforementioned research focused on the field of wireless video transmission and did not specifically target UAV video transmission. To address this, the authors in \cite{DBLP:journals/tcsv/LiuJ22} conducted cross-layer design research on UAV-based video streaming transmission. They proposed a quality of service (QoS) strategy that dynamically adjusted the transmission mode according to network load, latency, and packet loss rate. Nevertheless, they focused on the overall system design of multi-link parallel transmission, without delving into the issue of UAV video coding. Considering the dynamic nature and power constraints of UAVs, it is necessary to conduct joint resource allocation and control for UAV video coding and transmission under the constraints of end-to-end latency and power consumption. In this way, the QoE of video transmission at the receiving end can be improved.

\subsection{Motivation and Contributions}
Overall, the optimization of UAV location/trajectory and network resource allocation to enhance the performance of UAV video transmission has emerged as a hot research topic. However, the issue of joint source-channel coding control and resource optimization allocation for UAV video coding {and} transmission has not been adequately investigated. This paper proposes a {joint source-channel d-P-R-D} design scheme based on HEVC coding. The objective is to minimize end-to-end distortion and power consumption by optimizing video coding parameters and the allocation of UAV transmit power under the constraints of end-to-end latency, power budget, and source-channel rate matching.
The main research content and contributions of this paper are summarized as follows:

1) In the context of the HEVC standard and the UAV air-to-ground (AtG) channel propagation characteristics, eight-dimensional d-P-R-D models for UAV video coding and transmission have been formulated. Specifically, we explore a nonlinear regression method to establish a mathematical model describing the standard deviation of transformed residuals in HEVC video coding. We further deduce the d-P-R-D model for video coding with the nonlinear standard deviation model. Besides, a d-P-R-D model for UAV video transmission is constructed on the basis of AtG channel characteristics.

2) We formulate a joint source-channel optimization problem for UAV video coding and transmission. This problem is formulated on the basis of the eight-dimensional d-P-R-D models and a model capturing correlation between video coding and transmission. Its objective is to minimize end-to-end distortion and UAV power consumption by optimizing fine-grained parameters related to UAV video coding and transmission. 

3) The problem is confirmed to be a multi-objective and sequential-decision problem including a family of non-convex constraints, making it highly challenging to be solved. To solve this challenging problem, we develop an iterative solution algorithm based on Lyapunov optimization and convex approximation strategies. Initially, the sequential-decision problem is reduced in dimensionality to form a number of repeated optimization problems. Subsequently, for each repeated optimization problem, a two-stage iterative divide-and-conquer strategy is designed to decompose it into two independent sub-problems. Besides, a convex approximation strategy is explored to tackle non-convexity of sub-problems.

4) Finally, we design {objective and subjective} experiments to quantitatively analyze the stability of the proposed algorithm and the effectiveness of the joint source-channel optimization mechanism. We thoroughly discuss the impact of various parameters on the performance of the algorithm. Experimental results demonstrate that the proposed algorithm achieves mean-rate stability. Compared to benchmark algorithms, the proposed algorithm {achieves better video quality and stability performance, with a 47.74\% reduction in the variance of the obtained encoding bitrate.}

\section{system model and Problem formulation}
In this section, we describe the system model for the joint source-channel UAV video coding and transmission from three perspectives: the AtG channel model, the video coding d-P-R-D model, and the channel transmission d-P-R-D model. Based on these models, a joint source-channel optimization problem is formulated. 
\begin{figure}[!t]
\centering
\includegraphics[width=3.3 in, height = 1.1 in]{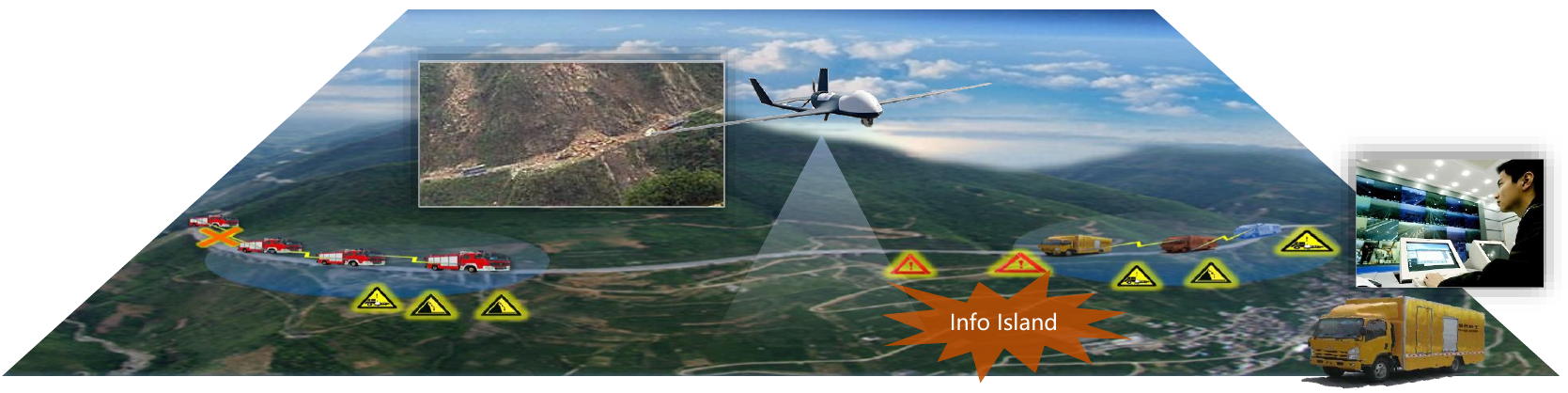}
\caption{A UAV video coding and transmission scenario.}
\label{fig_uav_video_scene}
\end{figure}
\subsection{AtG Channel Model}
This paper primarily considers the application scenarios where UAVs are deployed for agricultural monitoring, forest fire detection, and emergency search and rescue.
Particularly, the scenario investigated in this paper encompasses a UAV and an emergency communication vehicle (ECV). Given the potential for road traffic disruptions in areas such as mountainous and post-disaster zones, the ECV is deployed in proximity to the affected region to launch the UAV. The UAV is tasked with conducting surveillance and reconnaissance missions, capturing video footage for transmission, as illustrated in Fig. \ref{fig_uav_video_scene}. The captured video is encoded and compressed by the UAV and then promptly delivered to the ECV. To facilitate the construction of a mathematical model for the video transmission process, the temporal domain is discretized into a sequence of time slots, denoted by $t = \{ 1,2, \ldots \}$.
{Considering some practical application requirements}, the UAV is deployed to fly along a circular trajectory to extend the coverage range \cite{DBLP:journals/wcl/JiW23,DBLP:journals/tcom/DabiriHZKQ23,DBLP:journals/tnse/WuCYYWQ24}. The two-dimensional (2-D) horizontal coordinate of the UAV at time slot $t$ is represented as $\bm{q_{uav}}\left( t \right) = {[{x_{uav}}\left( t \right),{y_{uav}}\left( t \right)]^{\rm{T}}}$, and the position of the ECV is denoted as $\bm{q_{ecv}} = {[{x_{ecv}},{y_{ecv}}]^{\rm{T}}}$.
 Consider that the UAV flies at a constant altitude $H$, and the altitude of the ECV is negligible compared to $H$. The horizontal projection distance between the UAV and the ECV at time slot $t$ can be expressed as $dis(t) = \sqrt {{{({x_{uav}}\left( t \right) - {x_{ecv}})}^2} + {{({y_{uav}}\left( t \right) - {y_{ecv}})}^2}}$.

The UAV and the ECV communicate through an AtG link. In the AtG link, the UAV and the ECV typically establish line of sight (LoS) communication with a certain probability. This probability is contingent upon various factors including the flying altitude of the UAV, the horizontal projection distance between the UAV and the ECV, and the flying environment \cite{DBLP:journals/tnse/QinLZF24}. The commonly used expression for the probability of LoS communication is as follows
\begin{equation}\label{eq:los_prob}
{p_{{\rm{los}}}} = {\left[ {1 + a{e^{ - b(\theta (t) - a)}}} \right]^{ - 1}}
\end{equation}
where $a$ and $b$ are constant coefficients related to the flying environment of the UAV, and $\theta (t) = \frac{{180}}{\pi }\arctan (\frac{H}{{dis(t)}})$ represents the elevation angle between the ECV and the UAV. 

Disregarding the altitude of the ECV and the height of the UAV's antenna, the transmission loss of the AtG link \cite{DBLP:journals/tnse/QinLZF24} can be expressed as follows
\begin{equation}\label{eq:atg_los}
{L_{AtG}}(t) = 20{\log _{10}}(\sqrt {{H^2} + {{(dis(t))}^2}} ) + E{p_{{\mathop{\rm los}\nolimits} }} + F
\end{equation}
where $E = {\eta _{Los}} - {\eta _{NLos}}$, $F = 20{\log _{10}}(\frac{{4\pi {f_c}}}{c}) + {\eta _{NLos}}$, ${\eta _{Los}}$ and ${\eta _{NLos}}$ are environment-dependent, corresponding to transmission losses in LoS and non-line-of-sight (NLoS) links, respectively.

\subsection{Video Coding d-P-R-D Model}
This subsection constructs a video coding d-P-R-D model to abstract the problem, facilitating a systematic analysis of the interrelationships among the key elements such as delay, power, rate and distortion. This model provides a theoretical foundation for the optimization of video coding parameters such as quantization step and search range.

\subsubsection{Video Bitrate and Distortion Models}
The research presented in this paper is based on an advanced standard, particularly HEVC. Although HEVC continues to employ the traditional block-based hybrid coding framework, which includes key modules such as prediction, transformation, quantization, entropy coding, and loop filtering, it introduces innovative coding techniques in each of these modules. These techniques include quadtree-based block partitioning, up to $35$ intra-prediction modes, inter-frame variable-sized discrete cosine transform (DCT), context-based adaptive binary arithmetic coding (CABAC), and sample adaptive offset (SAO) for loop filter compensation. Compared to AVC, HEVC significantly improves compression efficiency, achieving an approximately $50\%$ reduction in bitrate while maintaining the same video quality, making it highly suitable for video transmission in bandwidth-constrained networks. {The VVC reference software VTM provides bit-rate savings of 44.4\% on average when compared to HEVC; however, with more than 10.2 times higher computational costs on average \cite{DBLP:conf/lascas/SiqueiraCG20,DBLP:journals/tcsv/KimDKK24}.
Considering the limited computational ability and power constraints of UAV, this paper constructs a d-P-R-D model for video coding based on HEVC.}

{Besides, this paper adopts the IPPPP coding mode without loss of generality\cite{DBLP:journals/tce/WangYX19}.} In the IPPPP coding mode, both encoding bitrate and distortion of P-frames in block-based hybrid video coding can be approximately represented as functions of the standard deviation of transformed residuals and the quantization step $Q$ \cite{DBLP:journals/tcsv/LiWX14}. Based on the research findings in literature\cite{DBLP:journals/tcsv/ChenW15}, the functional expression $\sigma $ can be obtained through nonlinear regression analysis. This paper selects standard video test sequences CITY (CIF, 352*288) and Coastguard (CIF, 352*288) and employs the HM-16.20 encoder to obtain encoding data samples. 
In the low-latency configuration file encoder\_lowdelay\_P\_main.cfg of HM, since quantization parameter $QP$, search range $\lambda $, and reference frames $\chi$ are independent parameters, their individual impacts on $\sigma $ can be assessed separately.
It should be noted that there is a one-to-one mapping relationship between the quantization parameter $QP$ mentioned in HM-16.20 and the quantization step $Q$ involved in the model. This paper first fixes $QP$ and $\chi$, and changes the value of $\lambda$ to obtain transformed residual data and calculate $\sigma $. Through nonlinear regression analysis, it has been found that $\sigma $ is approximately exponentially related to $\lambda$. Secondly, by fixing $\lambda$ and $\chi$, and changing $QP$, it can be concluded that $\sigma$ is approximately linearly related to $QP$.
Additionally, it is observed that the rate of change of $\sigma$ versus $\chi$ is much smaller than that of $\sigma$ versus $QP$, as well as $\sigma$ versus $\lambda$. It indicates that $\chi$ has a relatively smaller impact on $\sigma$ compared to $QP$ and $\lambda$. To reduce the computational complexity, this paper assigns a given value ${\chi _0} = 1$ to $\chi$\cite{DBLP:journals/tcsv/LiWX14}. Therefore, this paper adopts the following expression to calculate $\sigma$
\begin{equation}\label{eq:sigma_model}
\sigma (\lambda ,Q,\chi ) \approx \sigma (\lambda ,Q) = {a_1}{e^{ - {a_2}\lambda }} + {a_3} + {a_4}Q
\end{equation}
where the parameters ${a_1}$, ${a_2}$, ${a_3}$ and ${a_4}$ can be obtained through nonlinear regression analysis. Detailed parameter settings and experimental procedures will be elaborated in the experiment section.

When auxiliary information such as macroblock types is ignored, the bitrate of video coding can be approximated as the entropy of the transform-quantized coefficients\cite{DBLP:journals/tcsv/LiOHK09}. There are various assumptions for the distribution of transformed residuals. For instance, the research in \cite{DBLP:conf/vcip/WangLYGY22} assumed that transformed residuals followed a zero-mean generalized Gaussian distribution (GGD). However, the GGD has some limitations in practicality due to its multiple control parameters. In AVC, the Cauchy distribution has also been utilized, but the convergence issue of its mean and variance must be considered when employing the Cauchy distribution. The Laplace distribution achieves a good balance between computational complexity and model accuracy and has been widely applied \cite{DBLP:journals/tcsv/LiOHK09,DBLP:journals/tcsv/MaoWWK22}. This paper therefore chooses the Laplacian distribution as the statistical model for transformed residuals.

\begin{lemma}\label{lemma:lemma_video_coding_bitrate}
{\rm 
For {video} coding, where the transformed residuals follow a zero-mean Laplacian independent and identically distributed (i.i.d.) model, the relationship between the video coding bitrate ${R_e}(\lambda ,Q)$ and $Q$ as well as $\lambda$ can be mathematically represented as
\begin{equation}\label{eq:video_coding_bitrate}
\begin{array}{l}
{R_e}(\lambda ,Q) = 1 - ({e^{\frac{{\sqrt 2 Q(\mu  - 1)}}{{{a_1}{e^{ - {a_2}\lambda }} + {a_3} + {a_4}Q}}}} - 1) \times \\
{\log _2}2(1 - {e^{\frac{{\sqrt 2 Q(\mu  - 1)}}{{{a_1}{e^{ - {a_2}\lambda }} + {a_3} + {a_4}Q}}}}) + \\
{e^{\frac{{\sqrt 2 Q(\mu  - 1)}}{{{a_1}{e^{ - {a_2}\lambda }} + {a_3} + {a_4}Q}}}} \times \left[ {\frac{{\frac{{\sqrt 2 Q}}{{{a_1}{e^{ - {a_2}\lambda }} + {a_3} + {a_4}Q}}{{\log }_2}e}}{{1 - {e^{ - \frac{{\sqrt 2 Q}}{{{a_1}{e^{ - {a_2}\lambda }} + {a_3} + {a_4}Q}}}}}}} \right.\\
 - \left. {{{\log }_2}\left[ {{e^{ - \frac{{\sqrt 2 Q\mu }}{{{a_1}{e^{ - {a_2}\lambda }} + {a_3} + {a_4}Q}}}} - {e^{\frac{{\sqrt 2 Q(\mu  - 1)}}{{{a_1}{e^{ - {a_2}\lambda }} + {a_3} + {a_4}Q}}}}} \right]} \right]
\end{array}
\end{equation}}
\end{lemma}
\begin{proof}\renewcommand{\qedsymbol}{}
Please refer to Appendix A.
\end{proof}

It should be noted that, to model the relationship between video coding bitrate and time slot, the symbol ${R_e}(\lambda ,Q;t)$ is also utilized in the subsequent discussions to denote the video coding bitrate.
Compared to AVC, HEVC employs different RC strategies\cite{DBLP:journals/tcsv/LiLLC18}. In the research of HEVC, a hyperbolic model is commonly utilized to describe the relationship between video coding bitrate $R$ and distortion $D$\cite{DBLP:journals/tip/ChenP19}, i.e. $D(R) = C \cdot {R^{ - K}}$, where, $C$ and $K$ are parameters related to video content. After obtaining the video coding bitrate through source entropy, this paper models the relationship between video coding distortion and bitrate as follows
\begin{equation}\label{eq:video_coding_distortion}
{D_e}(R{}_e(\lambda ,Q;t)) = C \cdot {({R_e}(\lambda ,Q;t))^{ - K}}
\end{equation}

\subsubsection{Coding Delay and Power Models}
In the research field of video coding, encoding complexity is an important indicator of quantization encoding delay. Compared to AVC, HEVC employs flexible coding units (CU) partitioning based on a quadtree structure and a richer set of coding modes. It results in relatively higher computational complexity for CU recursive detection and coding mode selection. However, the CU partitioning and the rapid selection of coding modes have been extensively researched. As a result, the computational load of recursive detection and coding mode selection has been significantly reduced \cite{DBLP:journals/tcsv/StorchAZBP22,DBLP:journals/tmm/LinCYCKCL22}.

In the block-based hybrid video coding framework, ME has the highest computational complexity and is dominant throughout the entire encoding process. {Thus, the computational complexity of ME is utilized to approximate the complexity of the entire video coding process {\cite{DBLP:journals/tcsv/LiWX14}}. Further, in HEVC, to achieve higher compression efficiency, more complex inter-frame prediction modes and tree-structured motion compensation mechanisms are adopted compared with AVC. Therefore, this paper also employs motion estimation time (MET) to approximate the encoding latency of HEVC.}

The computational complexity of ME is primarily determined by $Q$ and the number of SAD operations required for each prediction unit, where $SAD = {(2\lambda  + 1)^2}\chi $. The MET for P-frames can be calculated by dividing the total number of CPU clock cycles consumed by executing SAD operations by the clock frequency (in Hz)\cite{DBLP:journals/tcsv/ChenW15}. Specifically, given $\chi  = 1$, the encoding delay for P-frames can be approximately calculated using the following formula
\begin{equation}\label{eq:video_coding_delay}
d(\lambda ,Q) = {{N{{(2\lambda  + 1)}^2}r(Q){C_s}}}/{{{F_{clk}}}}\end{equation}
where $N$ represents the total number of prediction units in a frame, ${(2\lambda  + 1)^2}$ denotes the number of SAD operations performed for each prediction unit across a three-dimensional search space, $r(Q)$ signifies the ratio of the actual number of SAD operations in the HM to the theoretical count of SAD operations, ${C_s}$ is the number of clock cycles required to complete one SAD operation on a given CPU, and ${F_{clk}}$ denotes the CPU's clock frequency. 

In CMOS circuits, the overall circuit power consumption is primarily composed of three components: static power consumption, dynamic power consumption, and short-term power consumption, which can be represented as {
\cite{DBLP:journals/tcsv/LiWX14,DBLP:9609444}}
\begin{equation}\label{eq:overall_circuit_power}
{P_e} = {C_L}V{V_{dd}}{F_{clk}} + {I_{sc}}{V_{dd}} + {I_{leakage}}{V_{dd}}
\end{equation}
where ${C_L}$ represents the capacitance, $V$ denotes the voltage fluctuation level, ${V_{dd}}$ signifies the circuit supply voltage, ${I_{sc}}$ refers to the short-circuit current, and ${I_{leakage}}$ indicates the leakage current. Due to the fact that the values of the latter two terms in  (\ref{eq:overall_circuit_power}) are typically very small and can be neglected, the equation can be simplified accordingly as
\begin{equation}\label{eq:simplified_overall_circuit_power}
{P_e} \approx {C_L}V{V_{dd}}{F_{clk}}
\end{equation}

In most cases, the voltage fluctuation level in CMOS circuits can reach up to ${V_{dd}}$. Thus, in (\ref{eq:simplified_overall_circuit_power}), $V$ can be replaced with ${V_{dd}}$. More importantly, when CMOS circuits operate under low-voltage conditions, there is a proportional relationship between the clock frequency and the supply voltage, meaning that ${V_{dd}}$ can be expressed as a function of ${F_{clk}}$. Considering that the battery on a UAV has a relatively low rated voltage, this paper models the encoding power consumption of the UAV as
\begin{equation}\label{eq:video_coding_power}
{P_e} = k \cdot F_{clk}^3
\end{equation}
where $k$ is a constant in the dynamic power scaling model, determined by the supply voltage and the effective switching capacitance of the circuit.

\subsection{Channel Transmission d-P-R-D Model}
In this subsection, a d-P-R-D model for UAV channel transmission is designed. This model serves as an innovative theoretical analysis tool for resource allocation and performance optimization of UAV video transmission.
\subsubsection{Rate and Distortion Models}
In the d-P-R-D model of channel transmission, the rate is defined as the data rate of the UAV communication channel. In the application scenario discussed in this paper, the SNR of the signal received by the ECV from the UAV at time slot can be expressed as
\begin{equation}\label{eq:signal_snr}
snr(t) = {{{P_t}(t)}}/({{{L_{AtG}}(t){P_n}}})
\end{equation}
where ${P_t}(t)$ represents the signal transmit power of the UAV at time slot $t$, and ${P_n}$ denotes the noise power. According to Shannon’s channel capacity formula, the data rate (in bps/Hz) received by the ECV at time slot $t$ can be expressed as
\begin{equation}\label{eq:data_rate}
{R_c}(t) = {\log _2}(1 + snr(t))
\end{equation}

In the d-P-R-D model of channel transmission, this paper focuses on the issue of bit error distortion in UAV communication channels. For the application scenario discussed in this paper, the video transmission between the UAV and the ECV employs a direct single-hop communication approach and avoids complex routing and relay equipment that may lead to extra distortion. Further, this paper imposes a constraint on the average bitrate of video coding, ensuring that it does not exceed the average data rate of the channel between the UAV and the ECV. This mechanism effectively eliminates packet loss distortion stemming from network congestion.

In the UAV deployment scenario considered in this paper, signals received by the ECV are predominantly LoS signals. Then, the bit error rate (BER) curve of signals received by the ECV approximates to that of an additive white Gaussian noise (AWGN) channel. Thus, the BER model of the AWGN channel is employed to approximate the BER of signals received by the ECV\cite{DBLP:journals/icl/AndradeRCZ22}, which can be expressed as below
\begin{equation}\label{eq:channel_distortion}
{D_c}({P_t}(t)) = \frac{1}{2}(1 - erf(\sqrt {{{{P_t}(t)}}/({{{L_{AtG}}(t){P_n}}})} ))
\end{equation}
where $erf(x)$ is the error function, and $erf(x) = \frac{2}{{\sqrt \pi  }}\int_0^x {{e^{ - {t^2}}}dt}$. It is evident from the aforementioned formula that increasing the transmit power of the UAV will lead to an increased SNR, thereby reducing the transmission distortion.

\subsubsection{Transmission Delay and Power Models}
Transmission delay ${t_{trans}}$ refers to the total time it takes for a data packet to travel from the sender of the channel, through the transmission process, and ultimately reach the receiver. It mainly includes sending delay ${t_{send}}$, propagation delay ${t_{propa}}$, 
and buffer processing delay ${t_{buffer}}$.

Sending delay ${t_{send}}$ is the time it takes for the UAV transmission module to send a data unit from the start to the end. It is calculated from the moment the first bit of the data unit is sent until the last bit of the data unit is completely transmitted. This delay is primarily determined by the length of data unit and channel capacity\cite{DBLP:journals/tnse/WuCYYWQ24}. Its calculation formula can be expressed as
\begin{equation}\label{eq:sending_delay}
{t_{send}} = {L}/({{B{R_c}(t)}})
\end{equation}
where $L$ represents the length of the sending data unit, $B$ denotes the UAV network bandwidth.

Propagation delay ${t_{propa}}$ is the time required for an electromagnetic signal to travel a certain distance within a communication channel. The calculation formula is ${t_{propa}} = \frac{d}{c}$, where $d$ represents the signal propagation distance between the UAV and the ECV, and $c$ denotes the propagation speed of the electromagnetic wave in the medium. In the application scenario discussed in this paper, since the propagation distance between the UAV and the ECV is relatively short, the propagation delay can be neglected.


Buffer processing delay ${t_{buffer}}$ refers to the time required for the receiver to perform operations such as error checking, data extraction, and buffer sorting after receiving a data packet. When designing the channel transmission d-P-R-D model, we focus on studying the behavior and performance of the sender. Therefore, buffer processing delay is outside the research scope of this paper and is temporarily not considered.

In the d-P-R-D model of channel transmission, the power consumption is equivalent to the UAV transmit power ${P_t}(t)$.

\subsection{Problem Formulation}
The research objective of this paper is to achieve the minimization of end-to-end distortion and total power consumption in the process of UAV video coding and transmission. Based on the aforementioned system model, this paper formulates a joint source-channel optimization problem. This problem aims to minimize the end-to-end video coding and transmission distortion and enhance energy efficiency through the joint optimization of UAV video coding parameters $\lambda $, $Q$, and the transmit power of the UAV under the constraints of delay, rate, and power. The formulated optimization problem is articulated as follows
\begin{subequations}\label{eq:original_problem}
\begin{alignat}{2}
&\mathop {\rm{Minimize}}\limits_{\lambda ,Q,{P_t}(t)} {D_e}({R_e}(\lambda ,Q;t)) + {\rho _1}{D_c}({P_t}(t)) + {\rho _2}{P_{tot}}(t)\\ \allowdisplaybreaks[4]
&{\rm s.t.}\mathop {\lim \inf }\limits_{t \to T} {{\bar R}_e}(\lambda ,Q;t) \le B{{\bar R}_c}(t)\\ \allowdisplaybreaks[4]
&{L}/({{B{d_{\max \_trans}}}}) \le {R_c}(t)\\
&d(\lambda ,Q) + {t_{send}} \le {d_{\max }}\\
&{P_{tot}}(t) = {P_c} + {P_e} + {P_t}(t) \le {P_{\max }}
\end{alignat}
\end{subequations}
where the time-averaged bitrate ${\bar R_e}(\lambda ,Q;t) = \frac{1}{t}\sum\limits_{\tau  = 1}^t {{R_e}(\lambda ,Q;\tau )}$ and the time-averaged data rate ${\bar R_c}(t) = \frac{1}{t}\sum\limits_{\tau  = 1}^t {{R_c}(\tau )}$, $T$ is total number of time slots, (\ref{eq:original_problem}b) is a rate causality constraint between source and channel that is enforced to avoid video playback rebuffering, (\ref{eq:original_problem}c) is the data rate constraint with ${d_{\max \_trans}}$ being the tolerable transmission delay, (\ref{eq:original_problem}d) represents the total delay constraint with ${d_{\max }}$ being the tolerable total delay, (\ref{eq:original_problem}e) is the total power constraint with ${P_{\max }}$ being the maximum UAV power and ${P_c}$  denoting the circuit power.
Besides, ${\rho _1}$ and ${\rho _2}$ are positive real numbers, and their values reflect the trade-off between source coding distortion, channel transmission bit error distortion, and the total power consumption of the UAV. The selection of ${\rho _1}$ and ${\rho _2}$ depends on user preferences as well as the Pareto boundary determined by the multi-objective optimization problem, as depicted in Fig. \ref{fig_pareto_frontiers}.
\begin{figure}[!t]
\centering
\includegraphics[width=2.6 in, height = 1.3 in]{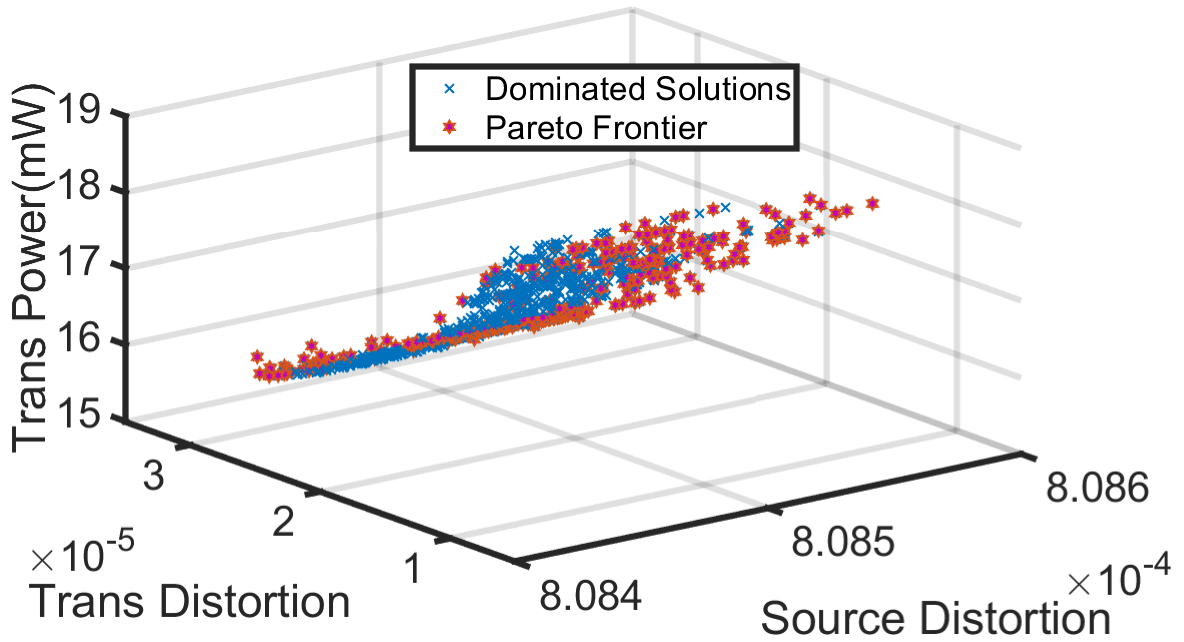}
\caption{Pareto frontiers and dominant solution sets.}
\label{fig_pareto_frontiers}
\end{figure}

The formulated problem is a multi-objective and non-convex optimization problem, and its solution is of high complexity. Firstly, the rate causality constraint is time-averaged. As time slot $t$ increases, the number of temporal decision variables grows exponentially. If traditional optimization methods are applied directly, the computational complexity is unacceptable. Secondly, the expression of the objective function is extremely complex. It includes fractional terms, exponential terms, and logarithmic terms. These terms are tightly coupled, further increasing the computational complexity of solving the problem. Additionally, the formulated multi-objective and non-convex optimization problem may be NP-hard, and it may even be impossible to find a global optimal solution.

To solve this challenging problem, this paper proposes a Lyapunov Repeated Iteration (LyaRI) algorithm. Initially, leveraging Lyapunov stability theory, LyaRI decouples the sequential-decision problem, breaking it down into multiple independent repeated optimization problems. It significantly reduces the computational complexity of solving the formulated problem. Subsequently, an iterative optimization strategy is employed to decompose the repeated optimization problem into two iterative optimization sub-problems. It achieves the decoupling of decision variables and makes sub-problems more solvable. Lastly, for the non-convex constraints within sub-problems, LyaRI adopts variable substitution and convex approximation strategies to approximately convert them into convex constraints.

\section{Lyapunov Repeated Iteration Algorithm Design}
\subsection{Problem Decomposition}
In addressing the constraints that include time-averaged terms, this paper employs Lyapunov drift-plus-penalty techniques for transforming. Specifically, a set of virtual queues $\left\{ {X\left( t \right)} \right\}$ is introduced and defined as follows
\begin{equation}\label{eq:virtual_queues}
X(t + 1) = X(t) + {R_e}(\lambda ,Q;t) - B{R_c}(t),\forall t
\end{equation}

In order to enforce the time-averaged constraint (\ref{eq:original_problem}b), the virtual queues need to meet the following stability conditions
\begin{equation}\label{eq:stability_conditions}
\mathop {\lim }\limits_{t \to \infty } {{E\left\{ {{{[X\left( t \right)]}^ + }} \right\}} \mathord{\left/
 {\vphantom {{E\left\{ {{{[X\left( t \right)]}^ + }} \right\}} t}} \right.
 \kern-\nulldelimiterspace} t} = 0,{[X\left( t \right)]^ + } = \max \{ X\left( t \right),0\} 
\end{equation}

Continuing, a Lyapunov function $L\left( t \right)$ is defined, which can be regarded as a scalar measure of the degree of constraint violation at a given time slot $t$. To simplify calculations, $L\left( t \right)$ is defined as $L\left( t \right) \buildrel \Delta \over = \frac{1}{2}{({[X(t)]^ + })^2}$. Accordingly, the expression for the Lyapunov drift-plus-penalty function can be represented as $\Delta \left( t \right) + V({D_e}({R_e}(\lambda ,Q;t)) + {\rho _1}{D_c}({P_t}(t)) + {\rho _2}{P_{tot}}(t))$, where $\Delta \left( t \right) = L\left( {t + 1} \right) - L\left( t \right)$ represents the Lyapunov drift, ${D_e}({R_e}(\lambda ,Q;t)) + {\rho _1}{D_c}({P_t}(t)) + {\rho _2}{P_{tot}}(t)$ is the penalty function of the optimization problem, and $V$ is a non-negative coefficient that characterizes the trade-off between constraint violation and optimality. Therefore, the solution to (\ref{eq:original_problem}) can be achieved by repeatedly minimizing the drift-plus-penalty function at each time slot while satisfying all non-time-averaged constraints within (\ref{eq:original_problem}).

\begin{lemma}\label{lemma:lemma_upper_bound}
{\rm 
At each time slot $t$, the upper bound of the Lyapunov drift-plus-penalty function can be expressed as follows
\begin{equation}\label{eq:upper_bound_Lyapunov}
\begin{array}{l}
\Delta \left( t \right) + V({D_e}({R_e}(\lambda ,Q;t)) + {\rho _1}{D_c}({P_t}(t)) + {\rho _2}{P_{tot}}(t))\\
 \le {[X(t)]^ + }\left( {{R_e}(\lambda ,Q;t) - B{R_c}(t)} \right) + \frac{1}{2}{(R_c^{\max })^2} + \\
V({D_e}({R_e}(\lambda ,Q;t)) + {\rho _1}{D_c}({P_t}(t)) + {\rho _2}{P_{tot}}(t))
\end{array}
\end{equation}}
\end{lemma}
\begin{proof}\renewcommand{\qedsymbol}{}
Please refer to Appendix B.
\end{proof}
(\ref{eq:upper_bound_Lyapunov}) indicates that the minimization of the drift-plus-penalty term can be approximated by minimizing its upper bound. As a result, the complex sequential-decision problem can be approximately transformed into a family of repeated optimization problems aiming at minimizing the upper bound of the drift-plus-penalty function. Specifically, we can repeatedly optimize the following problem at each time slot $t$ to obtain the upper bound of the objective function of (\ref{eq:Lyapunov_transformed_problem}). 
\begin{subequations}\label{eq:Lyapunov_transformed_problem}
\begin{alignat}{2}
&\mathop {\rm{Minimize}}\limits_{\lambda ,Q,{P_t}(t)} {[X(t)]^ + }\left( {{R_e}(\lambda ,Q;t) - B{R_c}(t)} \right) + \notag\\
&V({D_e}({R_e}(\lambda ,Q;t)) + {\rho _1}{D_c}({P_t}(t)) + {\rho _2}{P_{tot}}(t))\\
&{\rm s.t.} \text{ } \rm{Constraints \text{ } 
(\ref{eq:original_problem}c)-(\ref{eq:original_problem}e) \text{ } are \text{ } satisfied.}
\end{alignat}
\end{subequations}

However, solving problem (\ref{eq:Lyapunov_transformed_problem}) remains challenging due to the deep coupling between the quantization step $Q$ and the search range $\lambda $. To address this issue, this paper proposes an iterative optimization strategy, decomposing the problem into two sub-problems: quantization step optimization and power control along with search range optimization. 
\subsection{Solution to the Quantization Step Sub-Problem}
For any given transmit power ${P_t}(t)$ and search range $\lambda $, quantization step $Q$ can be optimized by solving the following sub-problem.
\begin{subequations}\label{eq:original_Q_subproblem}
\begin{alignat}{2}
&\mathop {\rm{Minimize}}\limits_Q {[X(i)]^ + }\left( {{R_e}(\lambda ,Q;t)} \right) + V  {D_e}({R_e}(\lambda ,Q;t))\\
&{\rm s.t.} \text{ } \rm{Constraint \text{ }(\ref{eq:original_problem}d)\text{ }is\text{ }satisfied.}
\end{alignat}
\end{subequations}

After substituting the R-D model into (\ref{eq:original_Q_subproblem}), it can be observed that the expression of the objective function is complex, making it difficult to analyze and optimize directly. To facilitate the analysis of the objective function of this sub-problem, this paper introduces a slack variable $\Omega $ and sets ${R_e}(\lambda ,Q;t) \le \Omega $. Meanwhile, a slack variable $\delta $ is introduced, and {we set $C \cdot {({R_e}(\lambda ,Q;t))^{ - K}} \leqslant \delta$.} Thus, (\ref{eq:original_Q_subproblem}) can be reformulated as
\begin{subequations}\label{eq:reformed_Q_subproblem}
\begin{alignat}{2}
&\mathop {\rm{Minimize}}\limits_Q {[X(i)]^ + }\Omega  + V\delta \\
&{\rm s.t.}{R_e}(\lambda ,Q;t) \le \Omega \\
&{R_e}(\lambda ,Q;t) \ge {({C}/{\delta })^{\frac{1}{K}}}\\
&\rm{Constraint \text{ }(\ref{eq:original_problem}d)\text{ }is\text{ }satisfied.}
\end{alignat}
\end{subequations}

Since (\ref{eq:reformed_Q_subproblem}b) is a non-convex constraint, this paper employs the successive convex approximation (SCA) strategy to convert (\ref{eq:reformed_Q_subproblem}b) into a convex constraint. Specifically, for any given local iterative point ${Q_0}$, we can have the following approximation
\begin{equation}\label{eq:video_bitrate_sca_Q0}
{R_e}({\lambda ^{(r)}},{Q_0};t) + \frac{{\partial {R_e}({\lambda ^{(r)}},{Q_0};t)}}{{\partial Q}}(Q - {Q_0}) \le \Omega 
\end{equation}

Moreover, (\ref{eq:reformed_Q_subproblem}c) is also non-convex and includes an exponential term. The following lemma elaborates on its approximation transformation.
\begin{lemma}\label{lemma:lemma_transform_convex}
{\rm 
By introducing auxiliary variables $\varepsilon $ and $\xi $, (\ref{eq:reformed_Q_subproblem}c) is approximately transformed into the following convex constraints
\begin{subequations}\label{eq:transform_convex_Q_subproblem}
\begin{alignat}{2}
&{R_e}({\lambda ^{(r)}},{Q_0};t) + \frac{{\partial {R_e}({\lambda ^{(r)}},{Q_0};t)}}{{\partial Q}}(Q - {Q_0}) \ge \varepsilon \\
&\left( {\varepsilon ,1,\xi } \right) \in {K_{\exp }}, \text{ } \left( {\delta ,C, - CK\xi } \right) \in {K_{\exp }}
\end{alignat}
\end{subequations}}
\end{lemma}
\begin{proof}\renewcommand{\qedsymbol}{}
Please refer to Appendix C.
\end{proof}

In summary, this paper introduces a set of slack variables $\Omega $, $\delta $,  $\varepsilon $, and  $\xi $, and employs the SCA strategy to perform convex approximations on the non-convex constraints. As a result, we can approximately transform (\ref{eq:original_Q_subproblem}) into the following convex optimization problem.
\begin{subequations}\label{eq:final_Q_subproblem}
\begin{alignat}{2}
&\mathop {\rm{Minimize}}\limits_Q {[X(i)]^ + }\Omega  + V\delta \\
&{\rm s.t.} \text{ } \rm{Constraints \text{ } (\ref{eq:original_problem}d), \text{ } 
(\ref{eq:video_bitrate_sca_Q0}), \text{ }
(\ref{eq:transform_convex_Q_subproblem}a), \text{ }(\ref{eq:transform_convex_Q_subproblem}b) \text{ } 
are \text{ } satisfied.}
\end{alignat}
\end{subequations}

The minimum value of the objective function in (\ref{eq:final_Q_subproblem}) provides an upper bound for (\ref{eq:original_Q_subproblem}). (\ref{eq:video_bitrate_sca_Q0}) and (\ref{eq:transform_convex_Q_subproblem}a) are linear constraints, while (\ref{eq:transform_convex_Q_subproblem}b) includes two exponential cone constraints. Thus, (\ref{eq:final_Q_subproblem}) is a convex conic optimization problem. Utilizing convex optimization tools, such as MOSEK, (\ref{eq:final_Q_subproblem}) can be effectively solved to obtain its optimal solution.

\subsection{Solution to the Power Control and Search Range Sub-Problem}
For any given quantization step $Q$, UAV transmit power ${P_t}(t)$ and search range $\lambda $ can be optimized by solving the following sub-problem.
\begin{subequations}\label{eq:original_lambda_Pt_subproblem}
\begin{alignat}{2}
&\mathop {\rm{Minimize}}\limits_{\lambda ,{P_t}(t)} {[X(t)]^ + }\left( {{R_e}(\lambda ,Q;t) - B{R_c}(t)} \right) + \notag \\
&V({D_e}({R_e}(\lambda ,Q;t)) + {\rho _1}{D_c}({P_t}(t)) + {\rho _2}{P_{tot}}(t))\\
&{\rm s.t.} \text{ } \rm{Constraints \text{ } 
(\ref{eq:original_problem}c)-(\ref{eq:original_problem}e) \text{ } are \text{ } satisfied.}
\end{alignat}
\end{subequations}

Similarly, after substituting the R-D model, it can be observed that the expression of the objective function of (\ref{eq:original_lambda_Pt_subproblem}) is complex. To facilitate the analysis of the objective function of this problem, this paper introduces a set of slack variables $\Omega $, $\delta $, and $\zeta $, and sets ${R_e}(\lambda ,Q;t) \le \Omega $, $C \cdot {({R_e}(\lambda ,Q;t))^{ - K}} \le \delta $, and ${D_c}({P_t}(t)) \le \zeta $. 
Based on (\ref{eq:video_bitrate_sca_Q0}) and the findings of \textbf{Lemma} \ref{lemma:lemma_transform_convex}, the SCA strategy is employed to perform a convex approximation on the non-convex constraints. Let ${{\rm{d}}_{coe}}{\rm{ = }}\frac{{N \cdot r(Q){C_s}}}{{{F_{clk}}}} \approx {d_1}({d_2}{e^{ - {d_3}Q}} + {d_4})$, where ${d_1}$, ${d_2}$,  ${d_3}$, and ${d_4}$ are nonlinear regression coefficients \cite{DBLP:journals/tcsv/LiWX14}.
Then, (\ref{eq:original_problem}d) can be expressed as ${(2\lambda  + 1)^2}{d_{coe}} + {t_{send}} \le {d_{\max }}$. Additionally, for the term ${R_c}(t)$ in the objective function, we replace it with a slack variable $\varphi $ and set ${R_c}(t) \ge \varphi $. Therefore, (\ref{eq:original_lambda_Pt_subproblem}) can be reformulated as
\begin{subequations}\label{eq:reformed_lambda_Pt_subproblem}
\begin{alignat}{2}
&\mathop {\rm{Minimize}}\limits_{\lambda ,{P_t}(t)} {[X(t)]^ + }\left( {\Omega  - B\varphi } \right) + V(\delta  + {\rho _1}\zeta  + {\rho _2}{P_{tot}}(t))\\
&{\rm s.t.}{R_e}({\lambda _0},{Q^{(r)}};t) + \frac{{\partial {R_e}({\lambda _0},{Q^{(r)}};t)}}{{\partial \lambda }}(\lambda  - {\lambda _0}) \le \Omega \\
&{R_e}({\lambda _0},{Q^{(r)}};t) + \frac{{\partial {R_e}({\lambda _0},{Q^{(r)}};t)}}{{\partial \lambda }}(\lambda  - {\lambda _0}) \ge \varepsilon \\
&{D_c}({P_t}({t_0})) + \frac{{d{D_c}({P_t}({t_0}))}}{{d{P_t}(t)}}({P_t}(t) - {P_t}({t_0})) \le \zeta \\
&{R_c}(t) \ge \varphi \\
&{(2\lambda  + 1)^2}  {d_{coe}} + {t_{send}} \le {d_{\max }}\\
&\rm{Constraints \text{ } (\ref{eq:original_problem}c), \text{ }(\ref{eq:original_problem}e), \text{ }
(\ref{eq:transform_convex_Q_subproblem}b)
 \text{ } are \text{ } satisfied.}
\end{alignat}
\end{subequations}
Regarding the constraints related to ${R_c}(t)$ in (\ref{eq:reformed_lambda_Pt_subproblem}), the following lemma outlines the specific transformation procedure.

\begin{lemma}\label{lemma:lemma_transform_data_rate}
{\rm 
The constraints related to ${R_c}(t)$ in (\ref{eq:reformed_lambda_Pt_subproblem}) can be equivalently transformed into the following forms
\begin{subequations}\label{eq:transform_data_rate_constraint}
\begin{alignat}{2}
&\left( {{{\rm Z}_1},1,\varphi \ln 2} \right) \in {K_{\exp }}\\
&\left( {\varphi ,\tau ,\sqrt {{{2L}}/{B}} } \right) \in Q_r^3, \text{ }
\left( {{1}/{2},{{\rm Z}_3},{{\rm Z}_2}} \right) \in Q_r^3\\
& - {{\rm Z}_1} + {{{P_t}(t)}}/({{{L_{AtG}}(t){P_n}}}) =  - 1, \ {{\rm Z}_2} = 2\lambda  + 1\\
&{{\rm Z}_3} = {({d_{\max }} - \tau )}/{{{d_{coe}}}}, \ \tau  \le {d_{\max \_trans}}
\end{alignat}
\end{subequations}}
\end{lemma}
\begin{proof}\renewcommand{\qedsymbol}{}
Please refer to Appendix D.
\end{proof}

In summary, this paper introduces a series of slack variables $\Omega $, $\delta $, $\zeta $, $\varepsilon $, $\xi $, $\varphi $, and $\tau $ to perform relaxation transformations on the optimization problem, and employs the SCA strategy to perform convex approximations on the non-convex constraints. (\ref{eq:original_lambda_Pt_subproblem}) can then be transformed into the following convex optimization problem
\begin{subequations}\label{eq:final_lambda_Pt_subproblem}
\begin{alignat}{2}
&\mathop {\rm{Minimize}}\limits_{\lambda ,{P_t}(t)} {[X(t)]^ + }\left( {\Omega  - B\varphi } \right) + V(\delta  + {\rho _1}\zeta  + {\rho _2}{P_{tot}}(t))\\
&{\rm s.t.} \ \rm{Constraints \text{ } 
(\ref{eq:original_problem}e), \text{ }
(\ref{eq:transform_convex_Q_subproblem}b), \text{ }
(\ref{eq:reformed_lambda_Pt_subproblem}b)-
(\ref{eq:reformed_lambda_Pt_subproblem}d), \text{ }} and \notag \\
&
\rm{(\ref{eq:transform_data_rate_constraint}a)-
(\ref{eq:transform_data_rate_constraint}d) \text{ } are \text{ } satisfied.}
\end{alignat}

\end{subequations}

The minimum value of the objective function in (\ref{eq:final_lambda_Pt_subproblem}) constitutes an upper bound for (\ref{eq:original_lambda_Pt_subproblem}). Within this problem, the first three and the last four constraints are linear, while the constraints from the fifth to the seventh are exponential cone constraints, and the eighth and ninth constraints are rotated cone constraints. Consequently, (\ref{eq:final_lambda_Pt_subproblem}) is a convex cone optimization problem, the optimal solution of which can be obtained using optimization tools such as MOSEK.

\subsection{Algorithm Design}
Based on the aforementioned theoretical analysis and derivations, we can summarize the main steps of solving (\ref{eq:original_problem}) in Algorithm \ref{alg:alg1}.

\begin{algorithm}
\caption{Lyapunov repeated iteration optimization algorithm, LyaRI.}
\label{alg:alg1}
\begin{algorithmic}[1]
\STATE \textbf{Initialization:} Initialize a positive random value for $X(1)$, set the maximum UAV {total power} ${P_{\max }}$, the maximum tolerable delay ${d_{\max }}$, the maximum number of time slots $T$, and the maximum number of iterations ${r_{\max }}$.
\FOR {each time slot $t = 1, 2, \ldots, T$}
\STATE Observe the virtual queue $X(t)$.
\STATE Initialize ${\lambda ^{(0)}}$, $P_t^{(0)}(t)$ and ${Q^{(0)}}$, and let $r = 0$.
\REPEAT
\STATE Given ${\lambda ^{(r)}}$ and $P_t^{(r)}(t)$, solve (\ref{eq:final_Q_subproblem}) to achieve the solution ${Q^{(r + 1)}}$.
\STATE Given ${Q^{(r + 1)}}$, ${\lambda ^{(r)}}$, and $P_t^{(r)}(t)$, solve (\ref{eq:final_lambda_Pt_subproblem}) to achieve the solution ${\lambda ^{(r + 1)}}$ and $P_t^{(r + 1)}(t)$.
\STATE Update $r = r + 1$
\UNTIL {Convergence or reach the maximum number of iteration $r_{max}$.}
\STATE Using (\ref{eq:virtual_queues}) to update $X(t + 1)$.
\ENDFOR
\end{algorithmic}
\end{algorithm}

At each time slot, the computational complexity of Algorithm \ref{alg:alg1} is primarily contributed by two components: {the optimization of the quantization step sub-problem and the optimization of transmit power and search range sub-problem.} These two sub-problems are transformed into convex problems and solved using an interior method, with their respective computational complexity being $O({N^{3.5}})$ and $O({M^{3.5}})$, where $N$ and $M$ are the dimensions of the decision variables of sub-problems. Additionally, the iterative optimization process of these two sub-problems needs to be considered. In the worst case, the total computational complexity of Algorithm \ref{alg:alg1} is $O({r_{\max }}(O({N^{3.5}}) + O({M^{3.5}})))$, where ${r_{\max }}$ represents the maximum number of iterations.

\begin{lemma}\label{lemma:lemma_algorithm_convergence}
{\rm 
Algorithm \ref{alg:alg1} is convergent, and the introduced virtual queue $X(t)$ is mean-rate stable.}
\end{lemma}
\begin{proof}\renewcommand{\qedsymbol}{}
Please refer to Appendix E.
\end{proof}

\section{Experiment and Result Analysis}
\subsection{Experimental Environment and Parameter Setting}
In this section, we will validate the effectiveness of the proposed algorithm. {The experiments utilize the latest version of official standard test model of HEVC (HM), namely HM 16.20.} The experiments select standard video sequences City (CIF, $352 \times 288$) and Coastguard (CIF, $352 \times 288$) for testing. The City sequence captures stationary targets by moving the camera position, while the Coastguard sequence tracks and shoots moving targets with a fixed camera. Both test sequences can effectively simulate typical UAV video surveillance scenarios.

In this paper, obtaining the regression coefficients through experimentation is a critical step. Initially, the {transformed} residual data of video sequences is acquired. The experiments utilize the low-latency configuration file encoder\_lowdelay\_P\_main.cfg of HM, with search range $\lambda  \in \{ 1,4,8,16,32\} $ and quantization parameter $QP \in \{ 18,24,30,36,42\} $. The encoding is performed on the basis of different $\lambda $ and $QP$.
During the encoding process, the transformed residual values of each prediction unit in P frames are written to files. A total of $25$ transformed residual data files can be obtained, corresponding to each $(\lambda ,QP)$ pair. Subsequently, for each transformed residual data file, the standard deviation of residuals for all prediction units is calculated on a per-frame basis.
To better capture video characteristics and mitigate the adverse effects on the experimental results due to fluctuations in inter-frame {transformed residual} standard deviations, the average of the transformed residual standard deviations for the first $10$ frames of the video sequence is calculated. Through the above calculations, experimental values of the residual standard deviation corresponding to different combinations of $\lambda $ and $QP$ can be obtained. Finally, based on the previously calculated residual standard deviations, nonlinear regression analysis is performed to determine values of coefficients such as ${a_1}$, ${a_2}$,  ${a_3}$, and ${a_4}$ in (\ref{eq:sigma_model}), thereby obtaining the closed-form expression for $\sigma (\lambda ,Q)$.

The configuration of other parameters for the experiments is as follows: $V = 4$, ${\rho _1} = 2$, ${\rho _2} = 0.01$, $C = 0.0015$, $K = 0.55$, $T = 40$, $\mu  = 0.1$, maximum total power consumption ${P_{\max }} = 2000$ mW, data block length $L = 1$ Mb, encoding power-related constant $k=1.3\times 10^{-24}$ mW/Hz, clock frequency ${F_{clk}} = 1000$ MHz, and network bandwidth $B = 10$ MHz. {We let the frame rate be $30$ fps,  the position of the ECV be $\left[ {50 \ \rm{m},50 \ \rm{m}} \right]$, the UAV flight radius be $250$ m, the UAV trajectory center position be $\left[ {250 \ \rm{m},250 \ \rm{m},500 \ \rm{m}} \right]$, and the UAV flying speed be $20$ m/s.} Other parameters related to the AtG channel model can be referred to\cite{DBLP:journals/jsac/YangXGQCC21}.

\subsection{Performance Evaluation}
In this section, we design extensive experiments to validate the performance of LyaRI. These experiments include verification of the stability of the algorithm, 
joint optimization performance verification, 
an assessment of the d-P-R-D model, a comparative analysis of the optimization results of LyaRI, a comparative analysis of the performance of LyaRI and HM rate control (HM RC) algorithm, and a comparative analysis of subjective performance.
\subsubsection{\textbf{Stability of LyaRI}}
\begin{figure}[!t]
\centering
\includegraphics[width=3.2 in, height = 1.6 in]{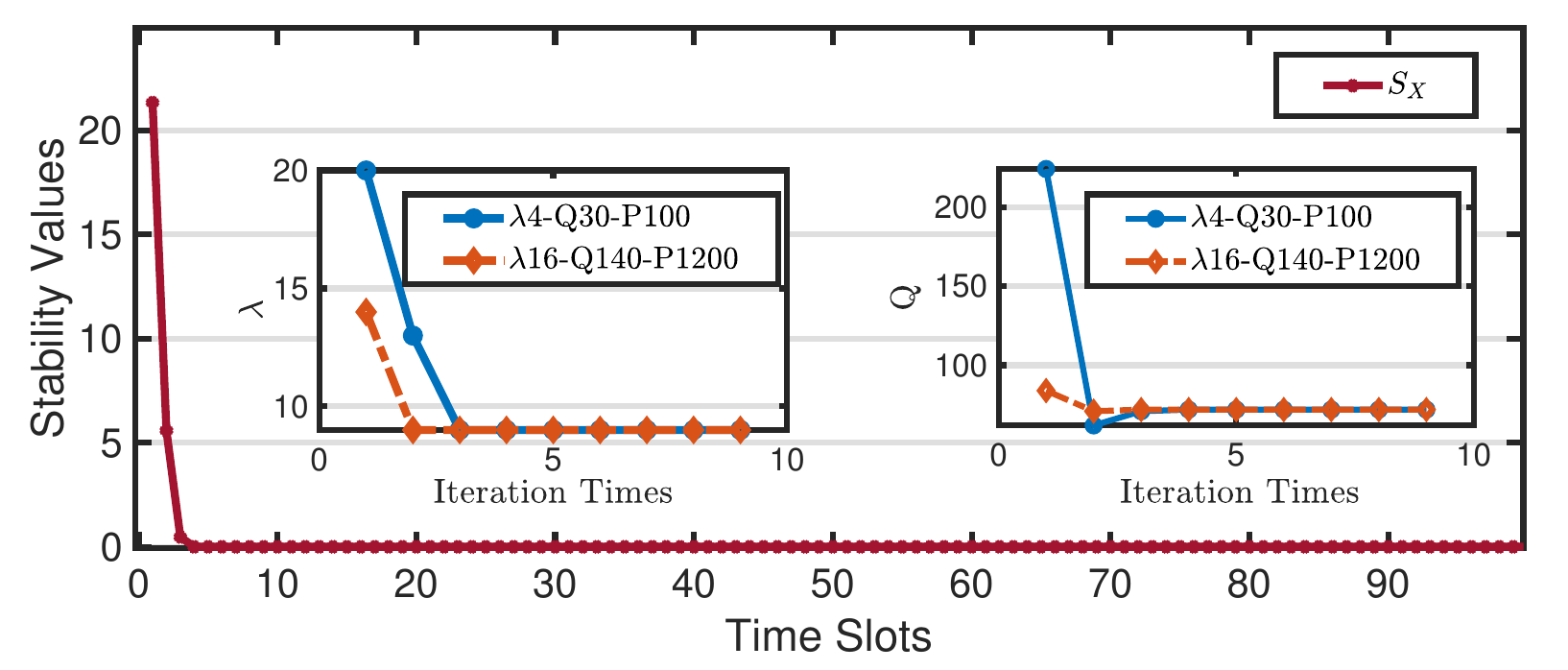}
\caption{Convergence performance of $\lambda $ and $Q$, and the stability of $X(t)$.}
\label{fig_convergence_stability}
\end{figure}
Fig. \ref{fig_convergence_stability} illustrates the convergence behavior of LyaRI on Coastguard sequence in terms of search range $\lambda $ and quantization step $Q$, as well as the stability of virtual queue $X(t)$. Here, we set ${d_{\max }} = 3$s, and set two combinations of different initial values to be ${\left( {\lambda,Q,{P_t}} \right)_{\min }} = (4,30,100 \ \rm{mW})$ and ${\left( {\lambda ,Q,{P_t}} \right)_{\max }} = (16,140,1200 \ \rm{mW})$, respectively. Experimental results demonstrate that both $\lambda $ and $Q$ can rapidly converge to their optimal values after a finite number of iterations. For different initial value combinations of $\lambda $, $Q$ and ${P_t}(t)$, although the number of iterations varies, the parameter ultimately converges to the same optimal value.

Additionally, we evaluate the stability of LyaRI, exactly the stability of the introduced virtual queue, which is defined as ${S_X} = {{\max {{[X\left( t \right)]}^ + }} \mathord{\left/
 {\vphantom {{\max {{[X\left( t \right)]}^ + }} t}} \right.
 \kern-\nulldelimiterspace} t}$. As shown in Fig. \ref{fig_convergence_stability}, the trend of virtual queue stability is plotted. It can be observed that the stability value of the virtual queue is bounded throughout the entire period of time and tends to zero as time slot $t$ increases. According to the definition of mean-rate stability, this virtual queue is mean-rate stable, which is consistent with the time-averaged constraint (\ref{eq:stability_conditions}).
\begin{figure}[!t]
\centering
  \subfigure[$\lambda $ and ${P_t}(t)$ vs. time slots]{
    \label{fig:subfig:a} 
    \includegraphics[width=1.8 in, height=1.3 in]{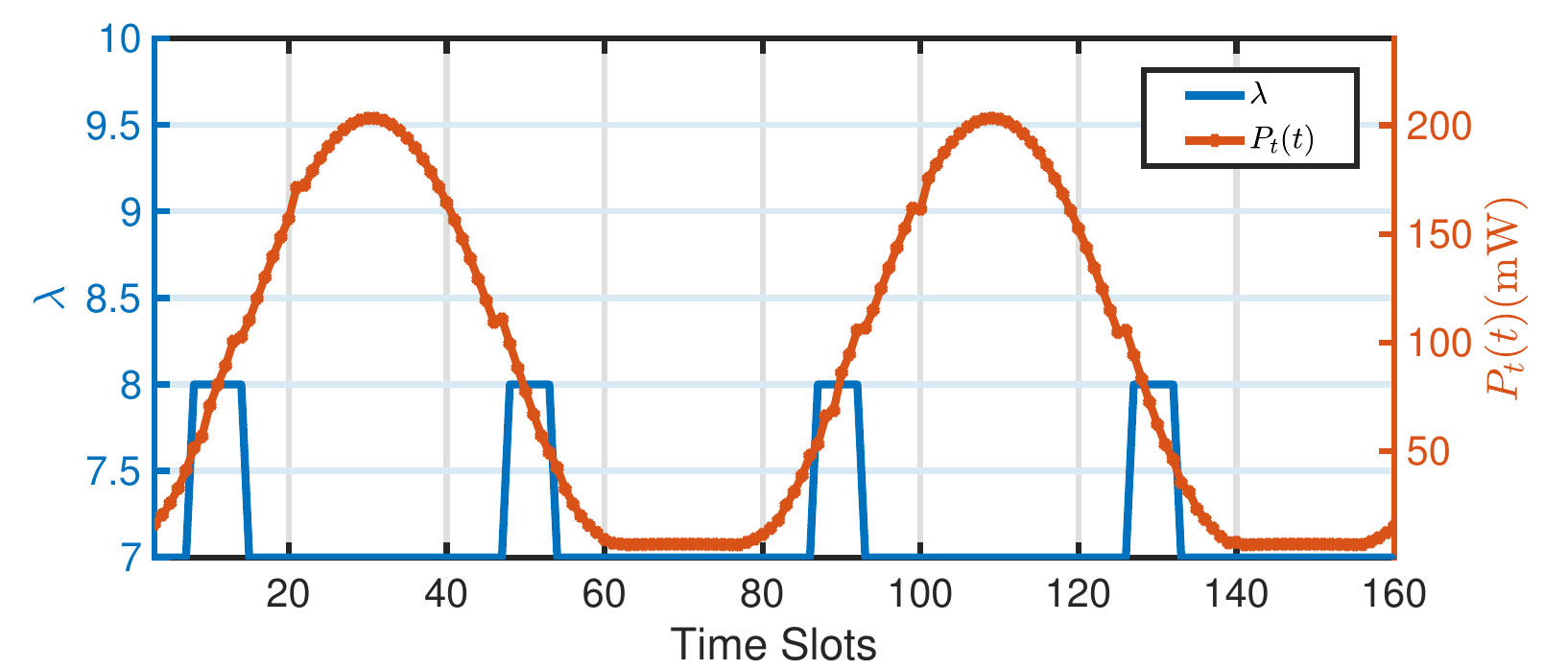}}
    \subfigure[UAV trajectory and key time slots annotations]{
    \label{fig:subfig:b} 
    \includegraphics[width=1.4 in]{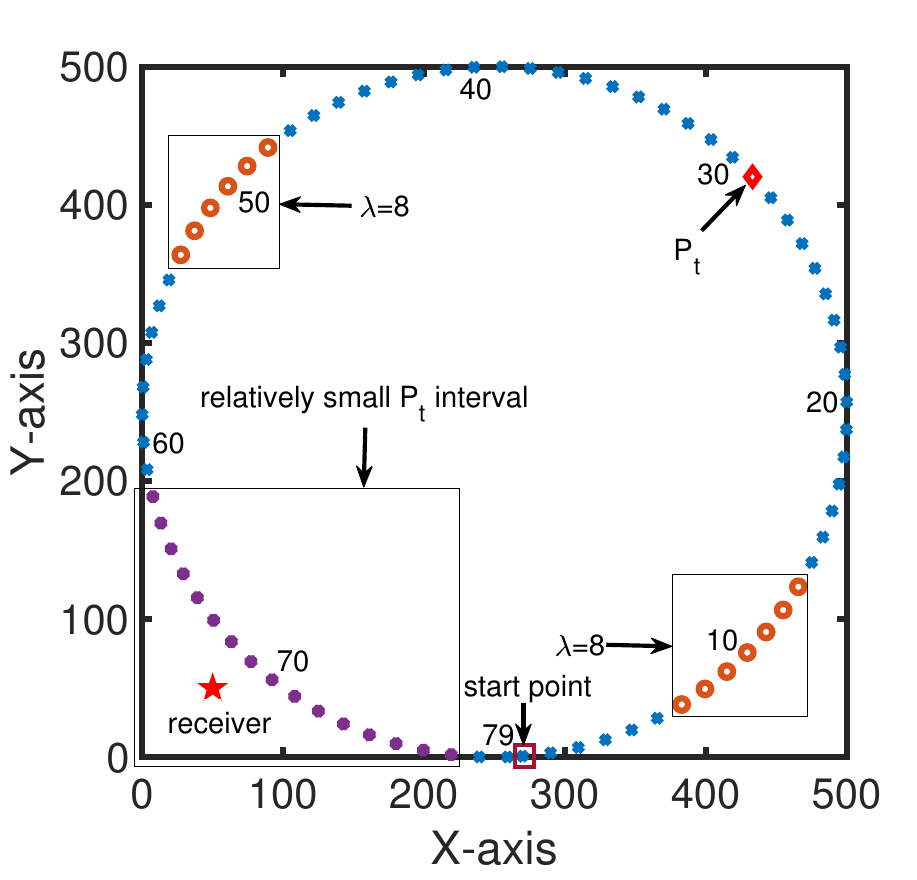}}
    \caption{$\lambda $ and ${P_t}(t)$ joint optimization performance analysis.}
\label{fig_joint_performance}
\end{figure}
\subsubsection{\textbf{Joint Optimization Performance Verification}}
Fig. \ref{fig_joint_performance}(a) presents experimental results on City sequence. It demonstrates the variation of $\lambda $ and ${P_t}(t)$ obtained by LyaRI when the UAV flies along a circular trajectory for two laps. Fig.  \ref{fig_joint_performance}(b) provides detailed annotations of the UAV's position, as well as $\lambda $ and ${P_t}(t)$ at regular and critical time slots during the UAV's first lap of flight.

From Fig. \ref{fig_joint_performance}, we can observe that: Firstly, $\lambda $ and ${P_t}(t)$ exhibit periodic changes and show an approximately symmetrical characteristic. This is because the distance between the UAV and the ECV changes periodically when the UAV completes a full circular trajectory flight.

Secondly, ${P_t}(t)$ is adjusted with the change in distance between the UAV and the ECV. To maintain the stability of data rate, ${P_t}(t)$ increases correspondingly when the distance between the UAV and the ECV increases, as shown in Fig. \ref{fig_joint_performance}(b). LyaRI achieves the maximum value of ${P_t}(t)$ at time slot $29$. Conversely, when the distance decreases, the transmit power ${P_t}(t)$ decreases accordingly. During the time slot intervals $\left[ {{\rm{62,77}}} \right]$ and $\left[ {{\rm{140,155}}} \right]$, as marked in Fig. \ref{fig_joint_performance}(a), ${P_t}(t)$ has relatively small values. It reflects the situation that the UAV is relatively close to the ECV.

Thirdly, $\lambda $ has a central value of $7$ after rounding, mainly determined by ${d_{\max }}$ and characteristics of video sequences. As marked in Fig. \ref{fig_joint_performance}(b), in the time slot intervals $\left[ {{\rm{7,13}}} \right]$ and $\left[ {{\rm{48,53}}} \right]$, $\lambda $ takes a value of $8$. It demonstrates an approximately symmetrical feature. Analysis reveals that these two time slot intervals correspond to the UAV-ECV distance changes most rapidly and also counterpart the great slope of the transmit power curve. Further, within these intervals, ${P_t}(t)$ increases more rapidly relative to the increase in distance, leading to an increase in data rate, and thus a decrease in transmission delay. While meeting the delay constraint, the available video coding delay increases, and LyaRI obtains a larger $\lambda $ to achieve better video coding quality.

In summary, under the changing UAV channel conditions, through the joint optimization of encoding and transmission parameters, the performance of video coding and transmission can be effectively improved. It demonstrates the value and significance of the joint source-channel optimization mechanism.

\begin{figure}[!t]
\centering
  \subfigure[Coastguard sequence]{
    \label{fig:subfig:a} 
    \includegraphics[width=2.4 in, height=1.4 in]{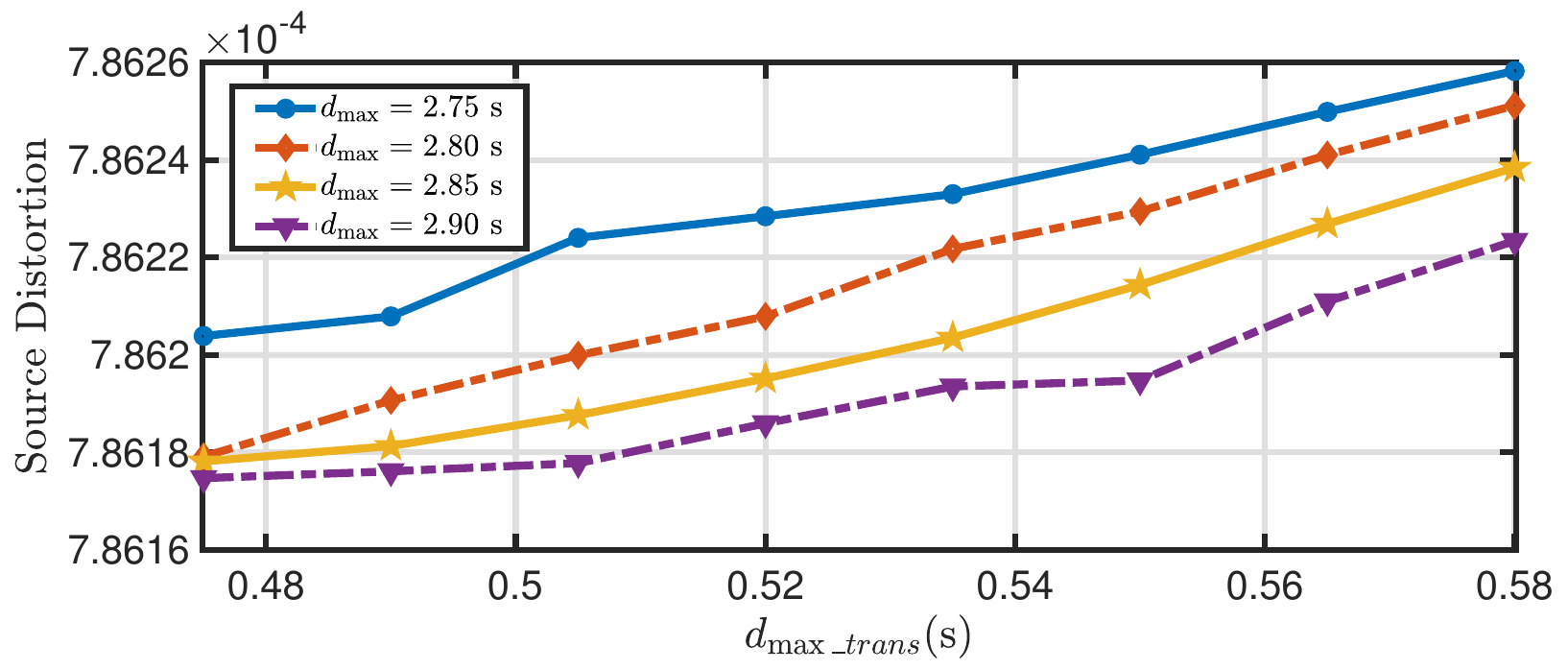}}
  \hspace{9pt}
    \subfigure[City sequence]{
    \label{fig:subfig:b} 
    \includegraphics[width=2.4 in, height=1.4 in]{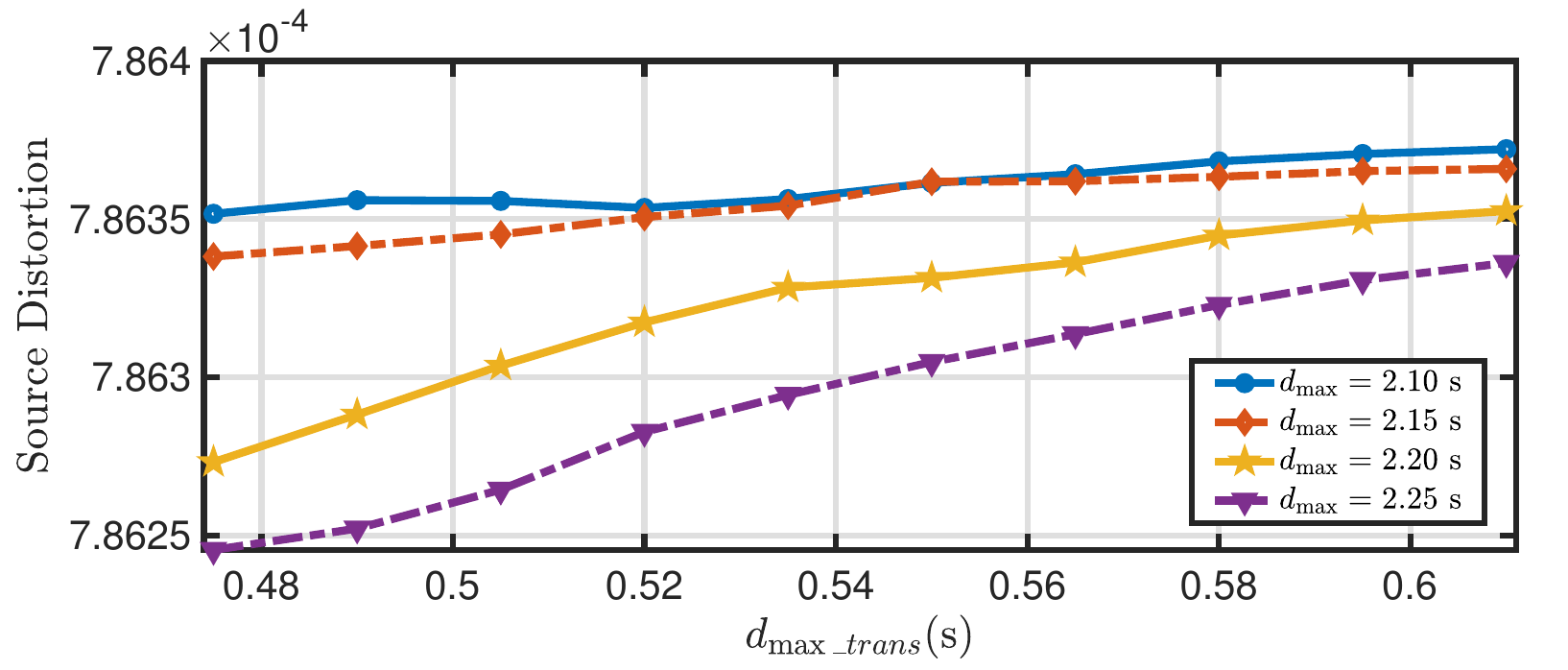}}
    \caption{Source distortion vs. ${d_{\max \_trans}}$ curves with different ${d_{\max }}$.}
\label{fig_model_analysis}
\end{figure}

\subsubsection{\textbf{d-P-R-D Model Analysis}}
(\ref{eq:original_problem}c) shows that the tolerable transmission delay ${d_{\max \_trans}}$ is inversely proportional to the UAV data rate. Therefore, this experiment adopts ${d_{\max \_trans}}$ to map the data rate. To verify the interrelationships among the elements in d-P-R-D model under the condition of a given maximum total power ${P_{\max }}$, this experiment obtains the source distortion (SD) curves with respect to ${d_{\max \_trans}}$ for different maximum delays ${d_{\max }}$.
\begin{figure*}[!t]
\centering
  \subfigure[Trends in Y-PSNR, varying QP and $\lambda$]{
    \label{fig:subfig:a} 
    \includegraphics[width= 2.0 in, height = 1.1 in]{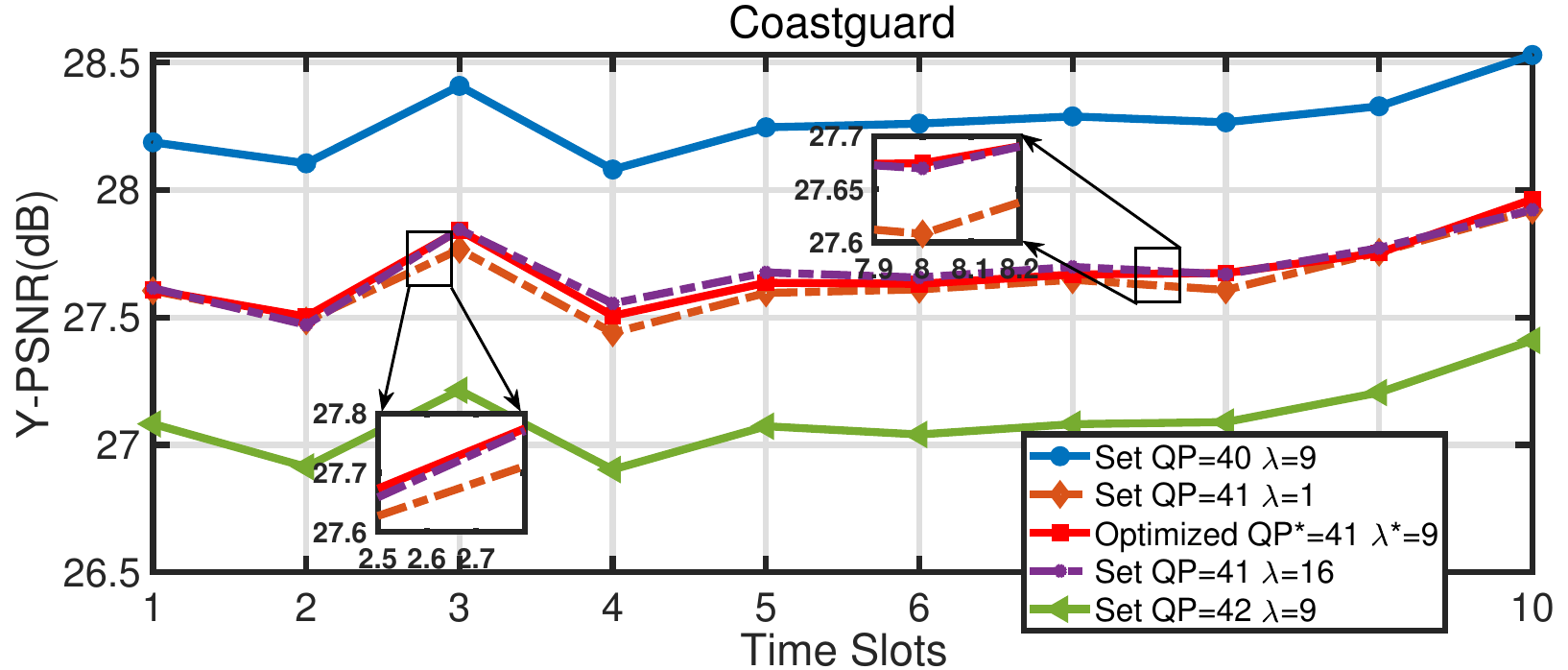}}
  \hspace{1pt}
    \subfigure[Trends in bitrate, varying QP and $\lambda$]{
    \label{fig:subfig:b} 
    \includegraphics[width= 2.0 in, height = 1.1 in]{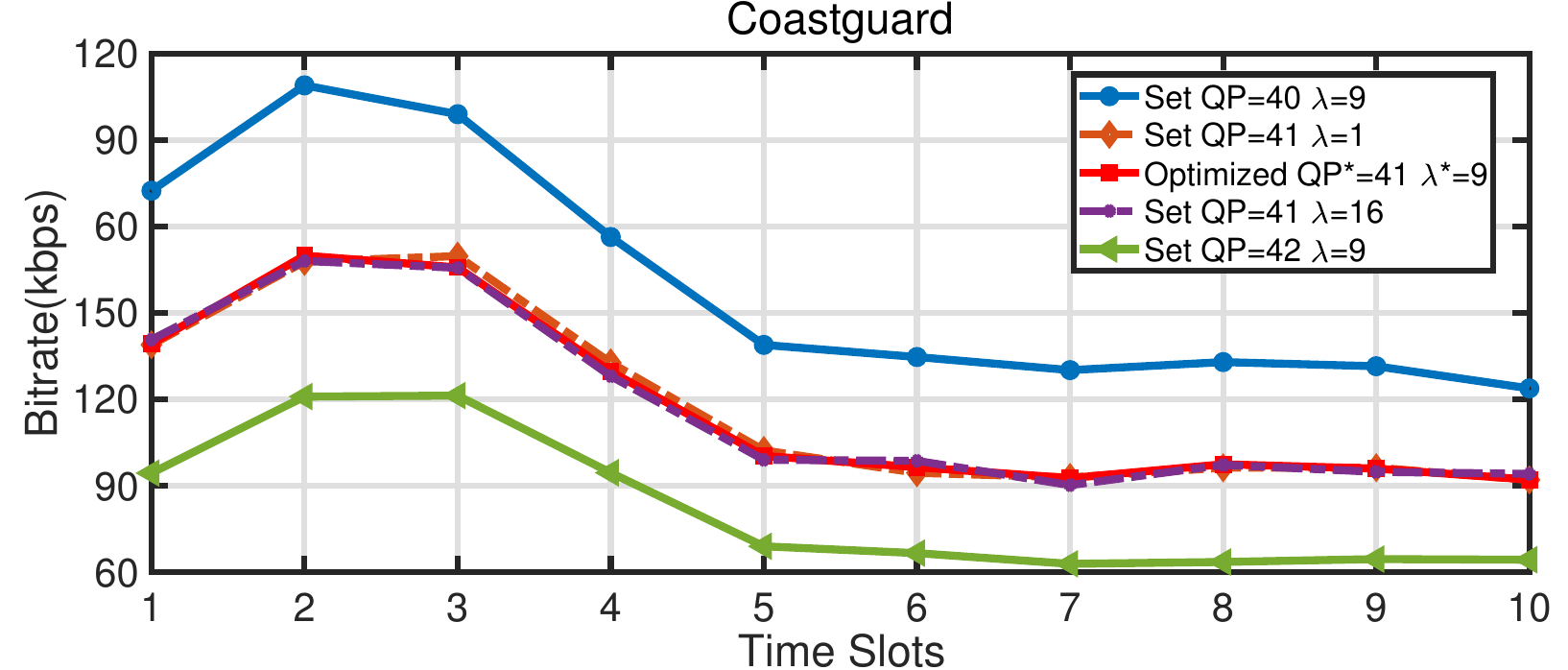}}
    \hspace{1pt}
    \subfigure[Trends in encode time, varying QP and $\lambda$]{
    \label{fig:subfig:c} 
    \includegraphics[width= 2.0 in, height = 1.1 in]{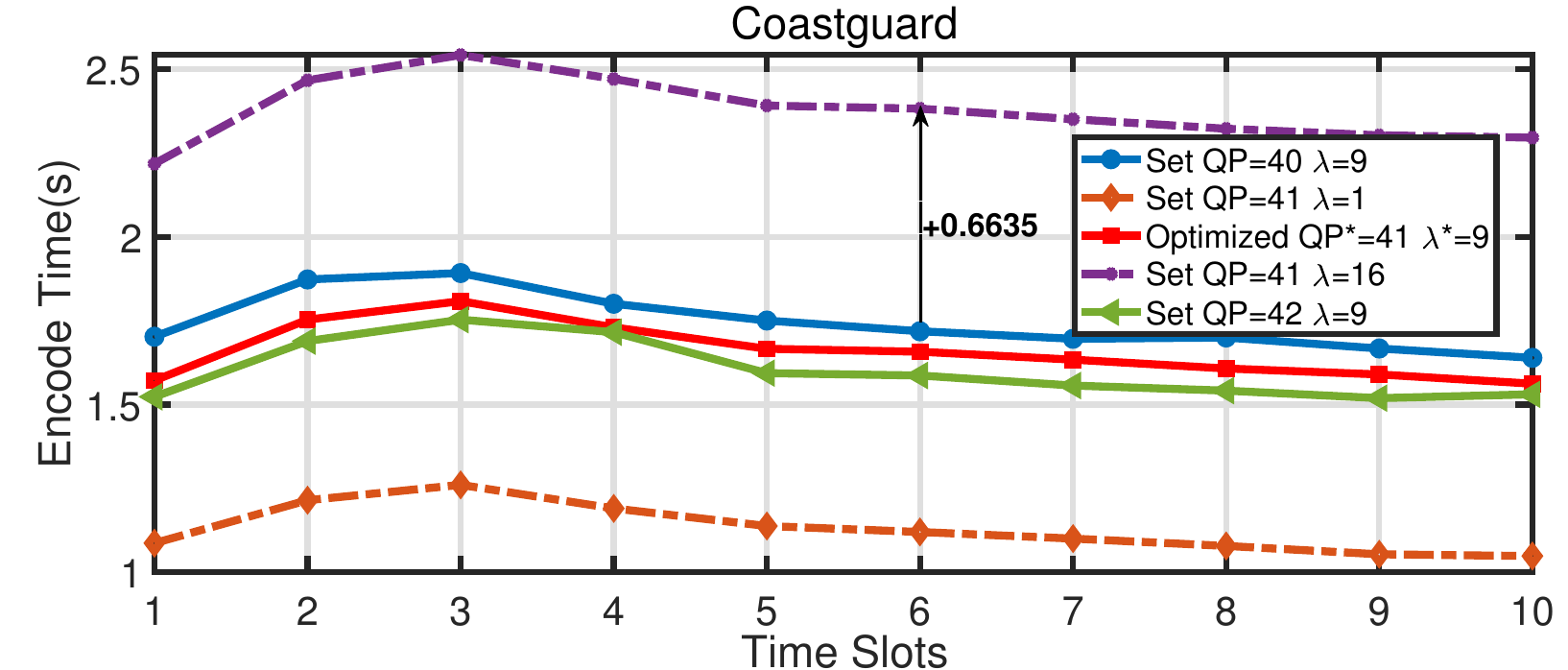}}
    \hspace{1pt}
    \subfigure[Trends in Y-PSNR, {varying QP and $\lambda$}]{
    \label{fig:subfig:d} 
    \includegraphics[width= 2.0 in, height = 1.1 in]{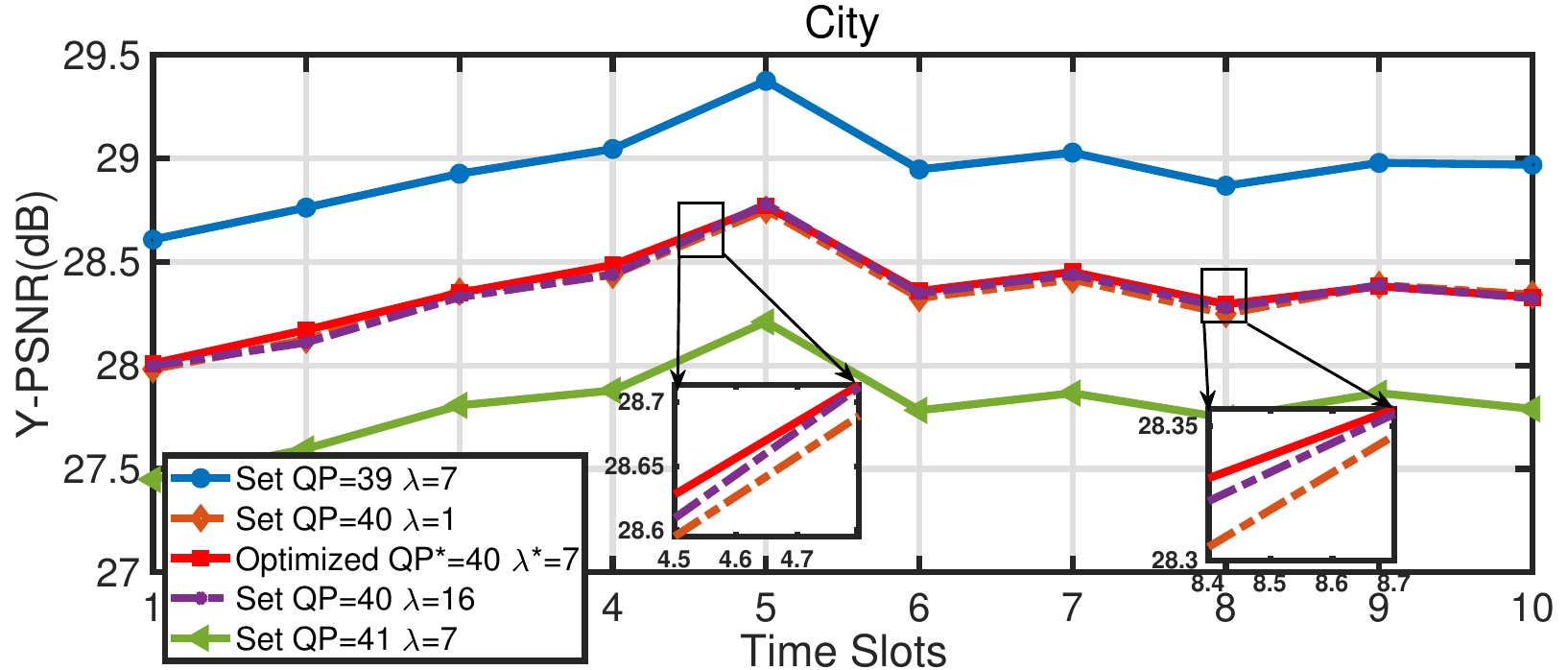}}
    \hspace{1pt}
    \subfigure[Trends in bitrate, varying QP and $\lambda$]{
    \label{fig:subfig:e} 
    \includegraphics[width= 2.0 in, height = 1.1 in]{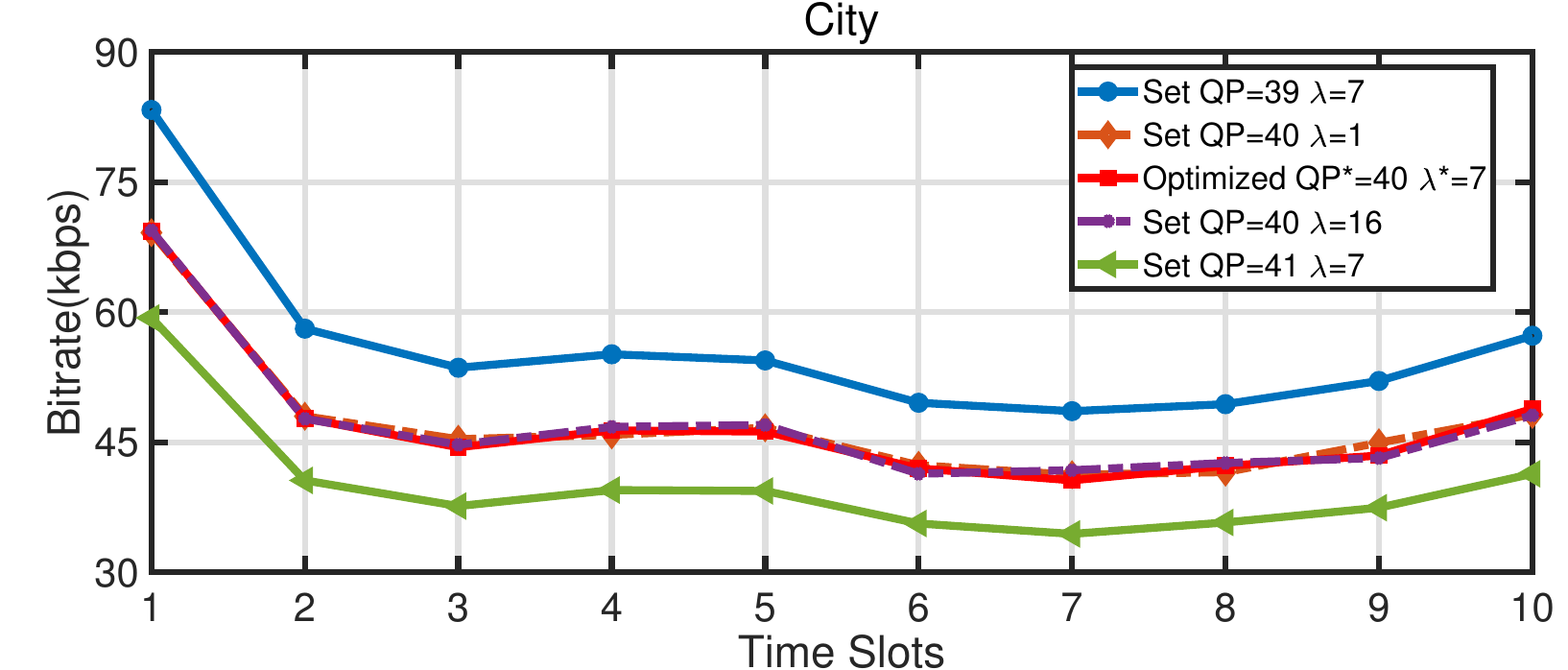}}
    \hspace{1pt}
    \subfigure[Trends in encode time, varying QP and $\lambda$]{
    \label{fig:subfig:f} 
    \includegraphics[width= 2.0 in, height = 1.1 in]{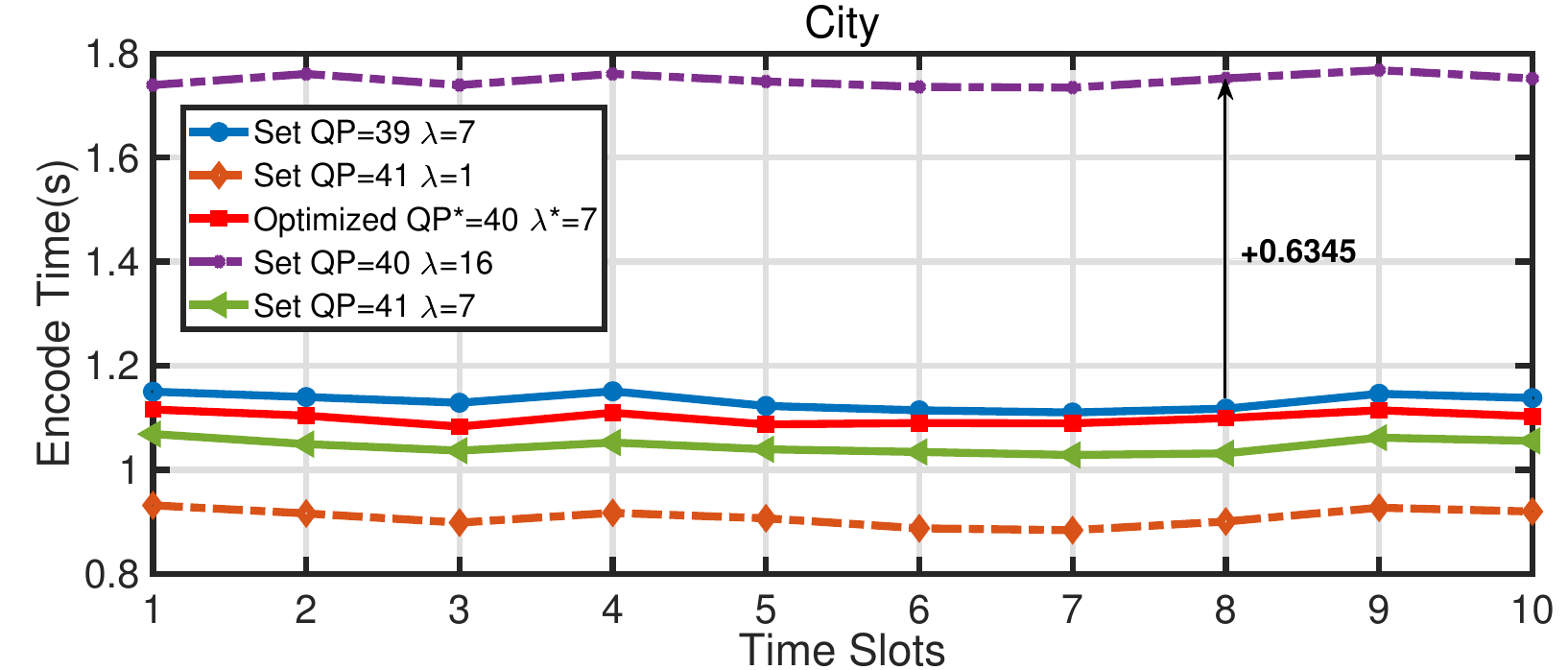}}
    \caption{Comparative analysis of encoding performance of the Coastguard and City sequences using LyaRI for the first $10$ time slots.}
\label{fig_performance_analysis_LyaRI}
\end{figure*}
For Coastguard sequence, we set ${d_{\max }} \in \left\{ {2.75,2.80,2.85,2.90} \right\}$ s. We choose a value for ${d_{\max }}$, vary the value of ${d_{\max \_trans}}$, and obtain the corresponding value of $SD$, thus acquiring the $SD - {d_{\max \_trans}}$ curve corresponding to ${d_{\max }}$. Fig. \ref{fig_model_analysis}(a) illustrates the $SD - {d_{\max \_trans}}$ curves for Coastguard sequence under different ${d_{\max }}$ settings. The observation results show that, for the same ${d_{\max }}$, the larger the ${d_{\max \_trans}}$, the greater the $SD$. This is because a larger ${d_{\max \_trans}}$ implies a relatively smaller data rate ${R_c}(t)$. According to (\ref{eq:original_problem}b), the source coding rate ${R_e}(\lambda ,Q;t)$ is also relatively smaller, thus leading to a larger $SD$. Besides, it can be observed that for the same ${d_{\max \_trans}}$, the larger the ${d_{\max }}$, the smaller the $SD$. This is because, for the same ${d_{\max \_trans}}$, a larger ${d_{\max }}$ implies a greater source encoding delay, and {thus} a smaller $SD$. Experimental results are consistent with the theoretical analysis of the video coding d-P-R-D model. It validates the effectiveness of the designed model.

For City sequence, we set ${d_{\max }} \in \left\{ {2.10,2.15,2.20,2.25} \right\}$ s. Fig. \ref{fig_model_analysis}(b) demonstrates that the trend of City sequence's curve is consistent with that of Coastguard sequence; hence, no detailed analysis is presented here.
\begin{figure*}[!t]
\centering
  \subfigure[Trends in Y-PSNR, varying initial QP]{
    \label{fig:subfig:a} 
    \includegraphics[width= 2.0 in, height = 1.1 in]{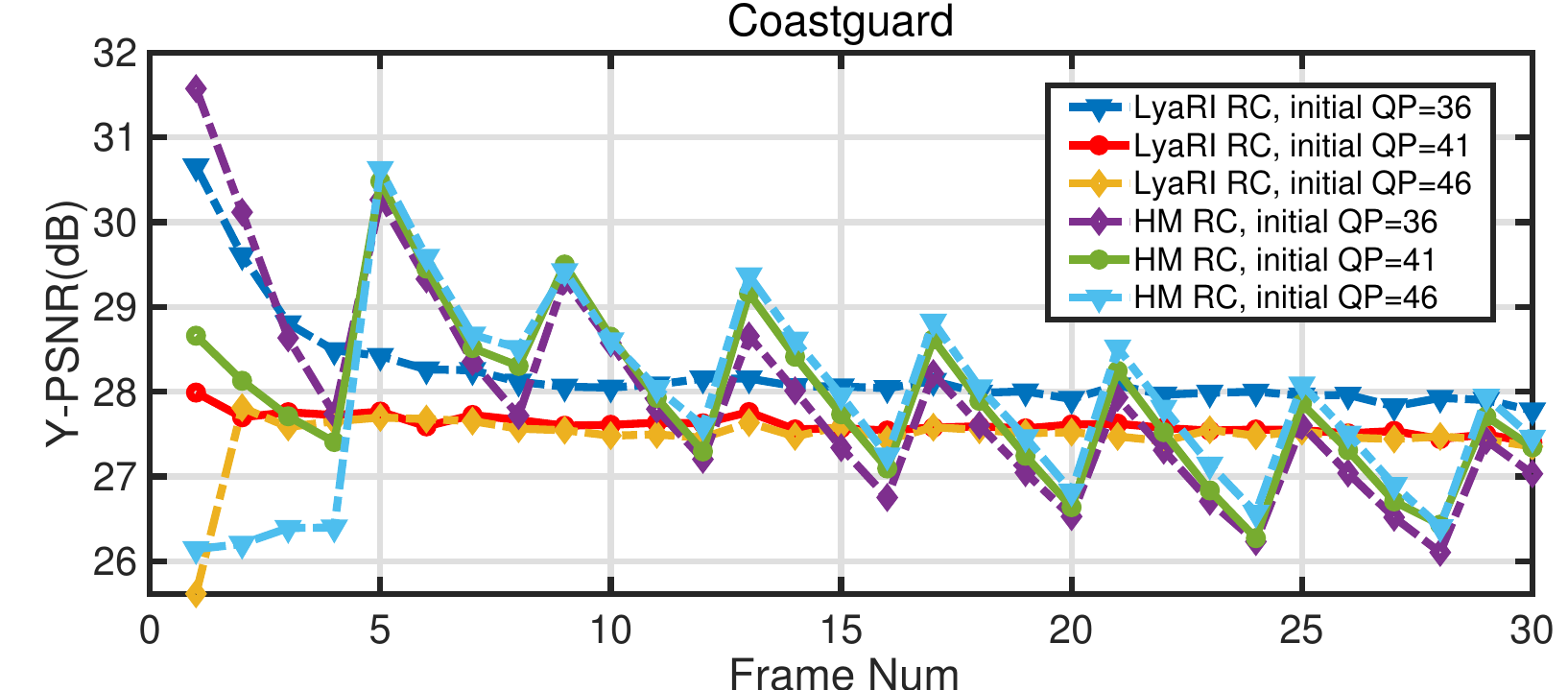}}
  \hspace{1pt}
    \subfigure[Trends in bit, varying initial QP]{
    \label{fig:subfig:b} 
    \includegraphics[width= 2.0 in, height = 1.1 in]{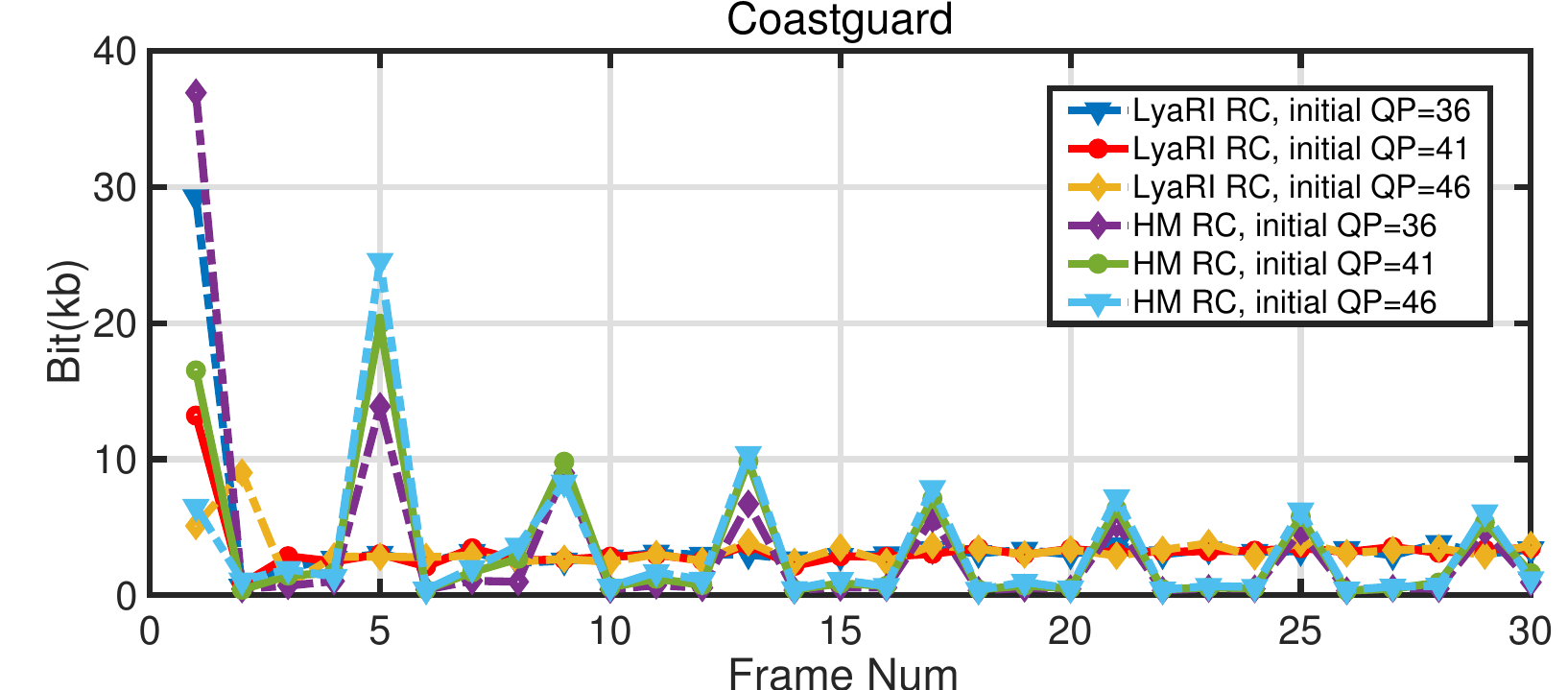}}
    \hspace{1pt}
    \subfigure[Trends in encode time, varying initial QP]{
    \label{fig:subfig:c} 
    \includegraphics[width= 2.0 in, height = 1.1 in]{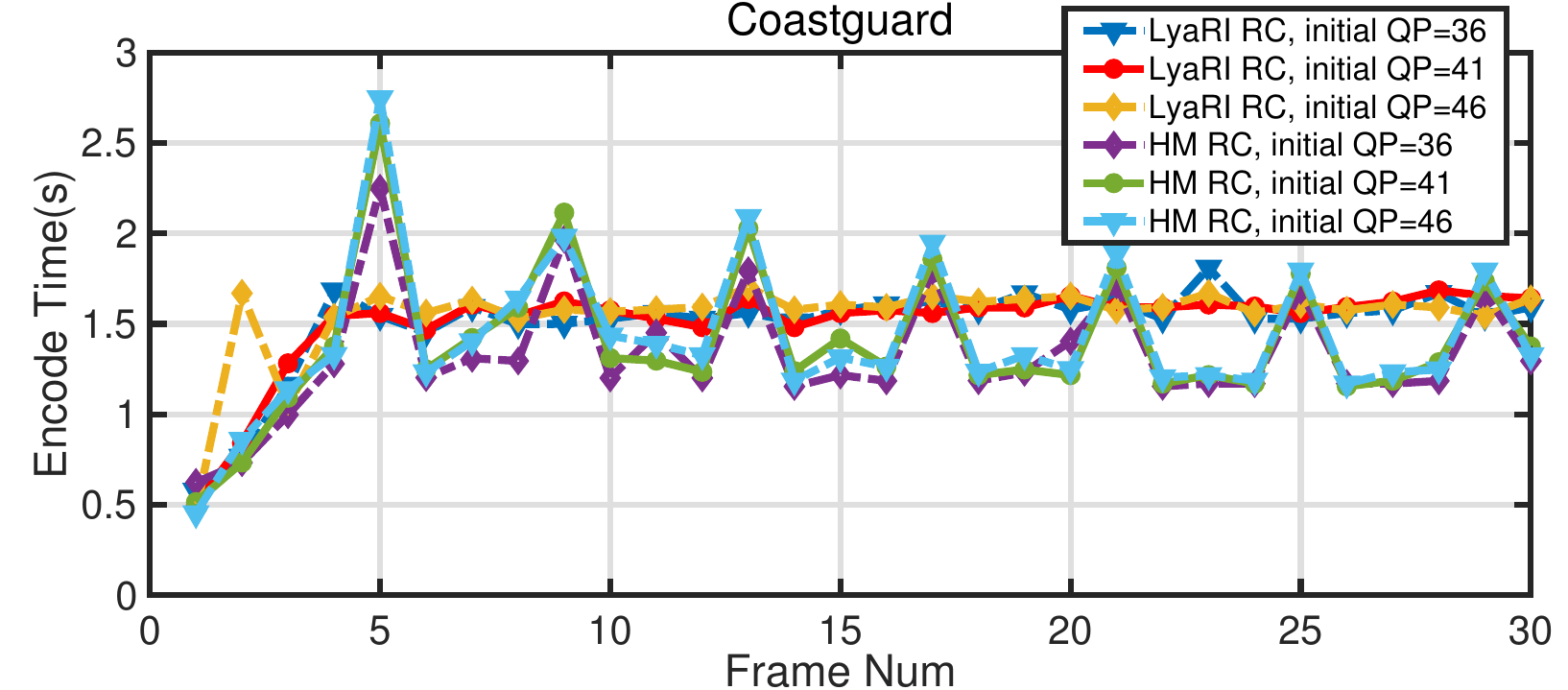}}
    \hspace{1pt}
    \subfigure[Trends in Y-PSNR{, varying initial QP}]{
    \label{fig:subfig:d} 
    \includegraphics[width= 2.0 in, height = 1.1 in]{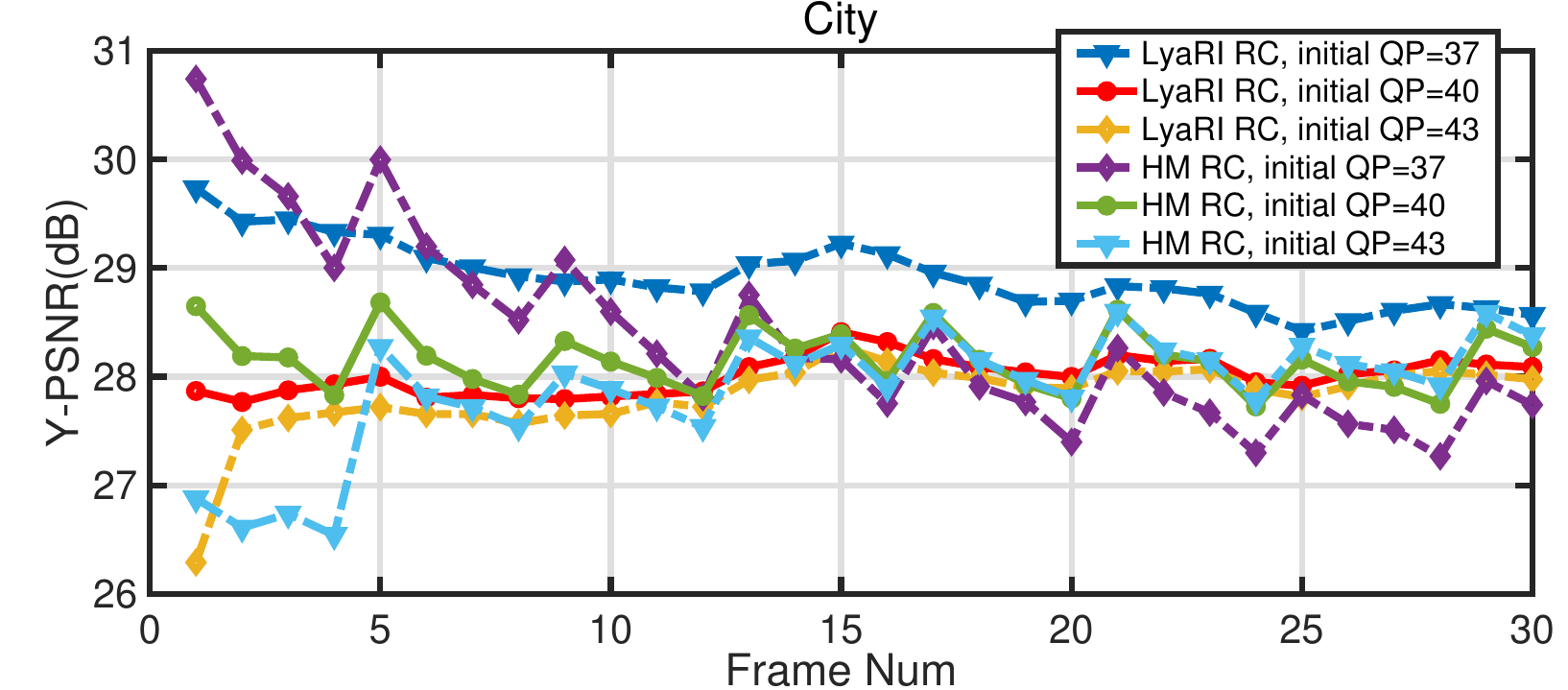}}
    \hspace{1pt}
    \subfigure[Trends in bit, varying initial QP]{
    \label{fig:subfig:e} 
    \includegraphics[width= 2.0 in, height = 1.1 in]{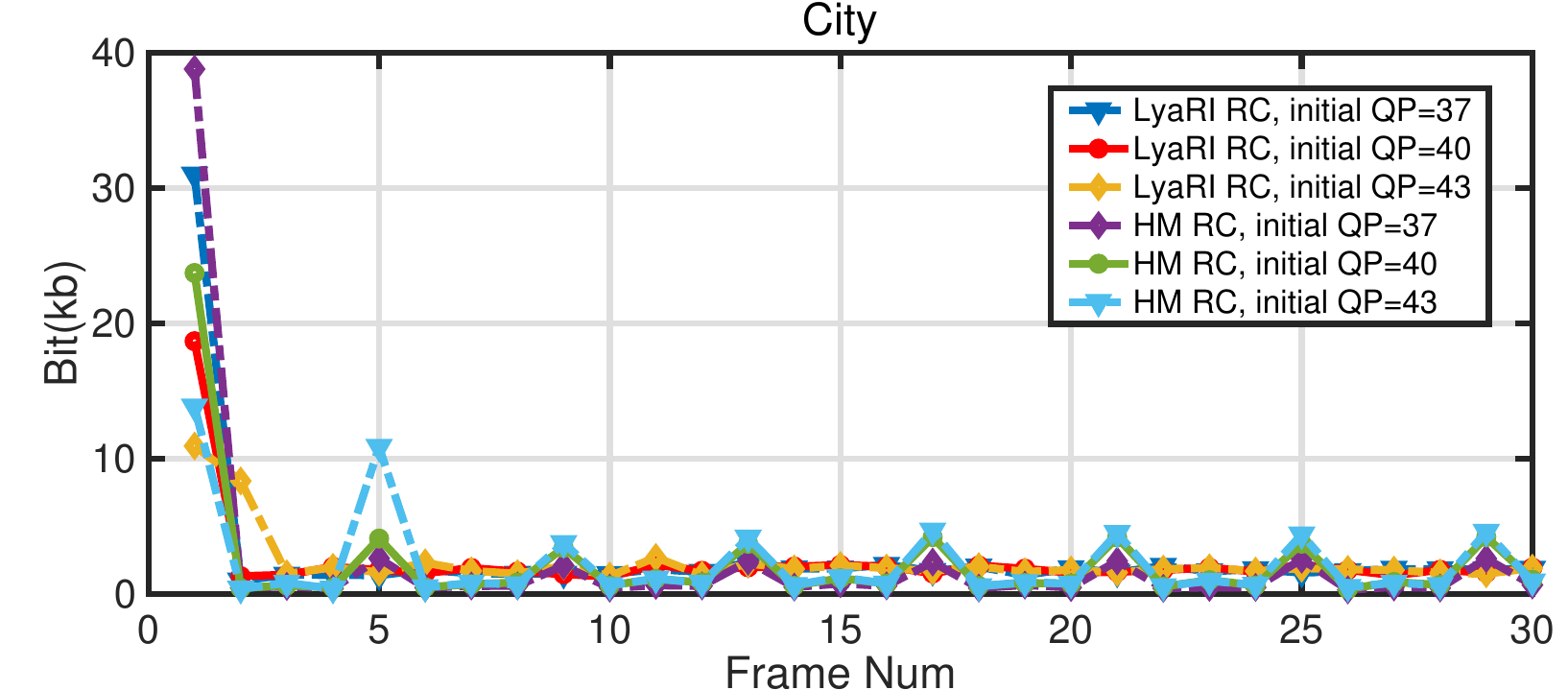}}
    \hspace{1pt}
    \subfigure[Trends in encode time, varying initial QP]{
    \label{fig:subfig:f} 
    \includegraphics[width= 2.0 in, height = 1.1 in]{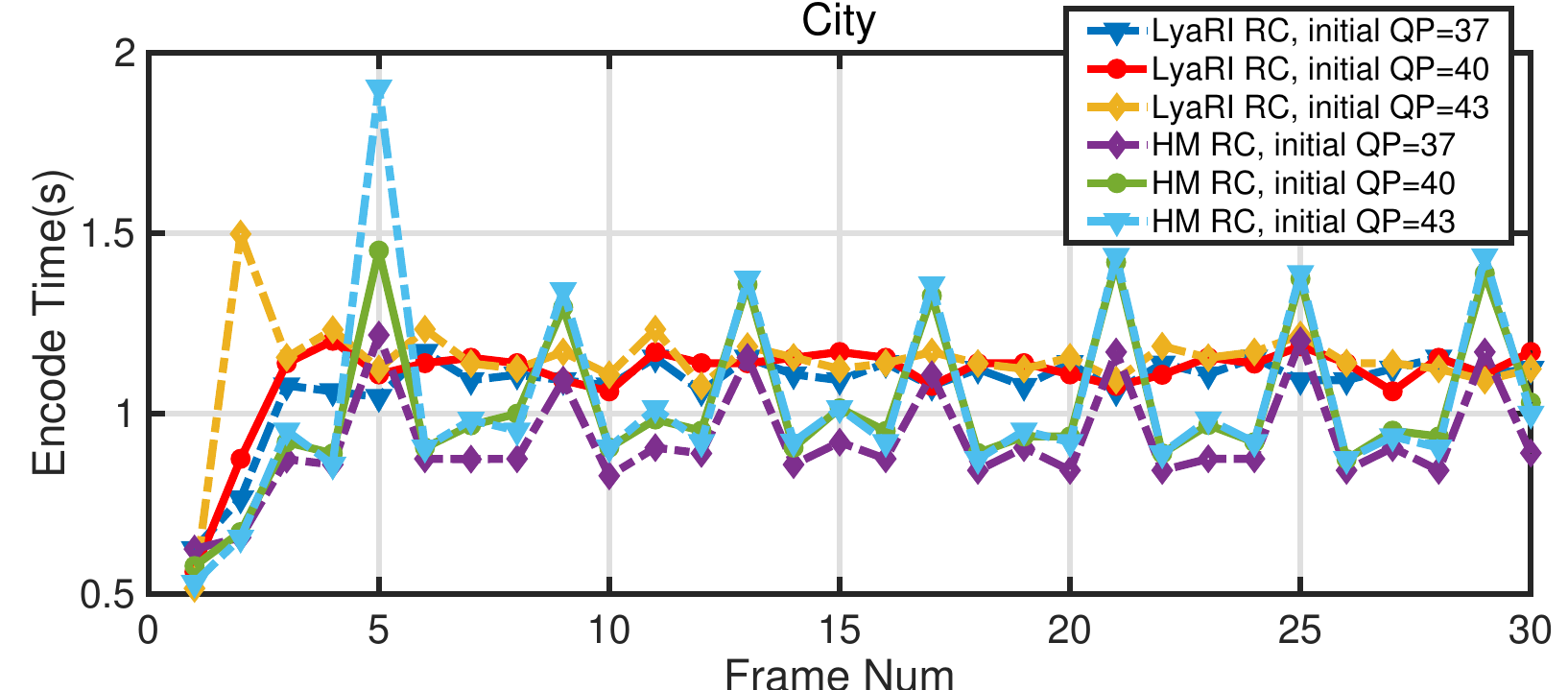}}
    \caption{Comparative analysis of encoding performance of the Coastguard and City sequences between LyaRI and the HM RC for the first 30 P-frames.}
\label{fig_comparative_LyaRI_HMRC}
\end{figure*}

\subsubsection{\textbf{{Comparative Analysis of LyaRI Optimization Results}}}
This experiment is primarily conducted to validate the encoding performance of LyaRI. Initially, the optimized video coding parameter pair $({\lambda ^ * },Q{P^*})$ is obtained by LyaRI. Subsequently, based on $({\lambda ^ * },Q{P^*})$, new encoding parameter pairs are configured by incrementing or decrementing their values. Thereafter, a comparative verification is conducted within HM by assessing Y-PSNR, encoding bitrate, and encoding time for each parameter pair during the first $10$ time slots.

For Coastguard sequence, we set $d_{\max}  = 3$ s, and the optimized parameter pair obtained by LyaRI is $({\lambda ^ * },Q{P^*}) = (9,41)$. With the optimized value ${\lambda ^ * } = 9$ being fixed, we change the value of $QP$. For instance, setting $(\lambda ,QP) = (9,40)$, it can be observed from Fig. \ref{fig_performance_analysis_LyaRI} that Y-PSNR is higher than the optimized solution. However, this type of parameter configuration violates the bitrate constraint. It indicates that the parameter configuration $(\lambda ,QP) = (9,40)$ is infeasible. On the other hand, when $QP$ is fixed at the optimized value of $41$, and $\lambda $ is increased, such as set $\lambda  = 16$. In this case, Y-PSNR and bitrate performance do not show significant improvement compared to the optimized solution. Nevertheless, the encoding time is much greater than that of the optimized solution.

The search range $\lambda $ significantly affects the complexity of video coding. After obtaining the optimized value of $\lambda $ using LyaRI, it eliminates the need to allocate valuable computational resources to additional ME. This is particularly beneficial in saving encoding time and reducing power consumption. Additionally, setting a lower $\lambda $, such as $\lambda  = 1$, results in a lower Y-PSNR value relative to the optimized solution. For City sequence, we set $d_{\max}  = 2$ s, and the optimized parameter pair obtained by LyaRI is $({\lambda ^ * },Q{P^*}) = (7,40)$. From Fig. \ref{fig_performance_analysis_LyaRI}, it can be observed that the encoding performance trend of City sequence is consistent with that of Coastguard sequence; hence, no further analysis is provided here.

\subsubsection{\textbf{Comparative Analysis with HM RC Algorithm}}
This experiment is designed to further validate the encoding performance of LyaRI by comparing it with HM RC. HM RC does not take into account encoding latency and power consumption. To ensure a fair comparison under constraints of maximum latency and power consumption, the search range $\lambda $ of HM RC is set to be the same as that of LyaRI in this experiment. Additionally, the same initial value of $QP$ is employed by LyaRI and HM RC.

Fig. \ref{fig_comparative_LyaRI_HMRC} presents a performance comparison between LyaRI and HM RC in terms of Y-PSNR, encoding bit, and encoding time for the first $30$ frames. The observations indicate that the performance of HM RC is significantly influenced by the initial value of $QP$. For HM RC, if the initial value of $QP$ is too low, the encoding bit of the initial few frames may substantially exceed the target bit. To meet the target average bit, it is necessary to increase the value of $QP$ for subsequent frames. It results in a relatively lower encoding bit for those frames. In contrast, the encoding bit obtained by LyaRI is less affected by the initial value of $QP$.

Fig. \ref{fig_comparative_LyaRI_HMRC} also reveals that HM RC exhibits significant fluctuations in Y-PSNR, encoding bit, and encoding time. For instance, when performing RC on the Coastguard sequence with $QP = 41$, the obtained variances of Y-PSNR, encoding bit, and encoding time are $0.9121$, $23.5987$, and $0.1605$, respectively. When performing RC on the City sequence with $QP = 40$, the variances of Y-PSNR, encoding bit, and encoding time are $0.0778$, $17.7689$, and $0.0460$, respectively. Although Y-PSNR values of some frames obtained by HM RC are higher than those of LyaRI, they violate the bitrate and latency constraints. In contrast, LyaRI is able to achieve smooth and stable performance in terms of Y-PSNR, encoding bitrate, and encoding time while satisfying constraint conditions. For instance, when performing RC on the Coastguard sequence with $QP = 41$, the variances of Y-PSNR, encoding bit, and encoding time are $0.0126$, $3.6917$, and $0.0579$, respectively. When performing RC on the City sequence with $QP = 40$, the variances of Y-PSNR, encoding bit, and encoding time are $0.0268$, $9.2853$, and $0.0135$, respectively. Further, encoding stability is crucial for a good user experience in video coding. The reasons for the encoding stability achieved by LyaRI are as follows: Firstly, the designed d-P-R-D model is derived based on the first few frames of the current group of pictures (GOP) in video sequences. The adoption of d-P-R-D model in LyaRI allows for the determination of respective encoding parameter pairs $({\lambda ^ * },Q{P^*})$ for different video sequences. Secondly, for each GOP, the encoding parameter pair remains constant, which contributes a lot to the stability of Y-PSNR and encoding bit across different frames. In contrast, HM RC adjusts the quantization parameter frame by frame to meet the target bit requirements. Consequently, Y-PSNR and encoding bit of HM RC fluctuate with the adjustment of the quantization parameter.

\begin{figure*}[htbp]
\centering
  \subfigure[Original: Frame index = 29]{
    \label{fig:subfig:a} 
    \includegraphics[width= 2.24 in, height = 1.3 in]{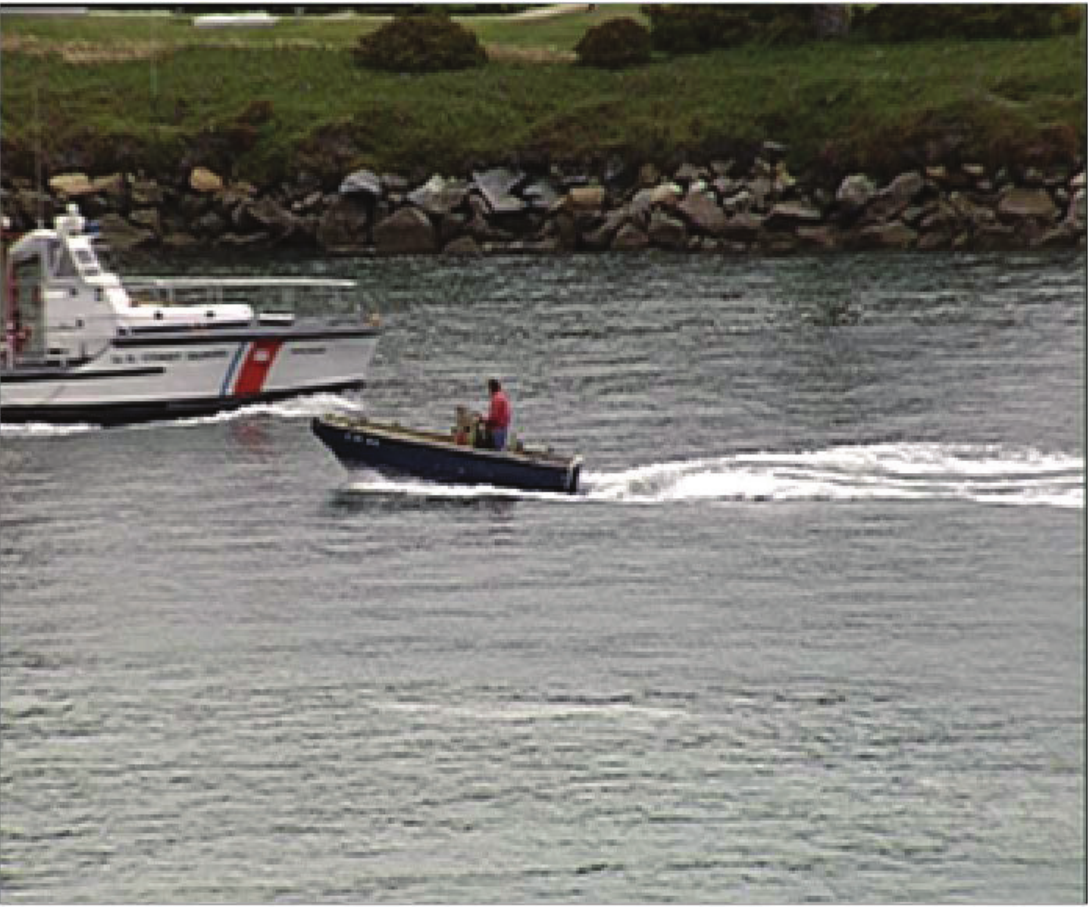}}
  \hspace{1pt}
    \subfigure[Original: Frame index = 30]{
    \label{fig:subfig:b} 
    \includegraphics[width= 2.24 in, height = 1.3 in]{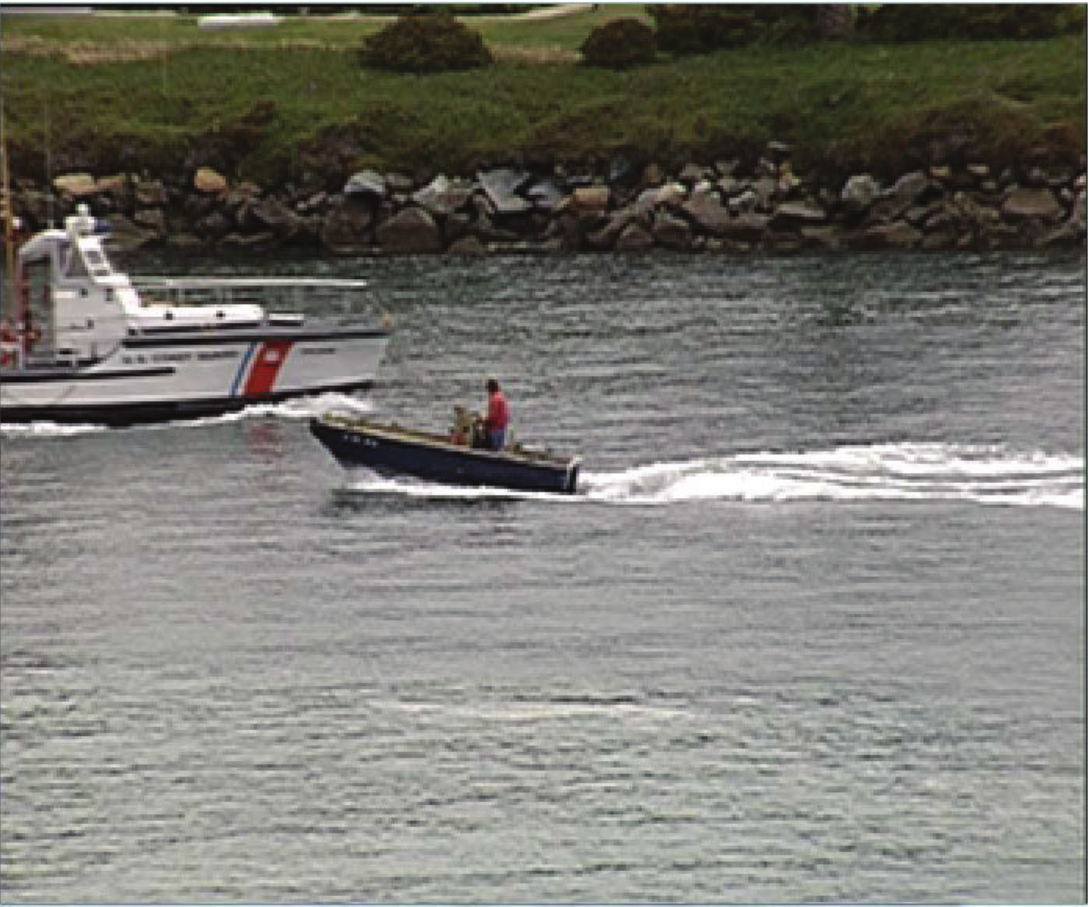}}
    \hspace{1pt}
    \subfigure[Original: Frame index = 48]{
    \label{fig:subfig:c} 
    \includegraphics[width= 2.24 in, height = 1.3 in]{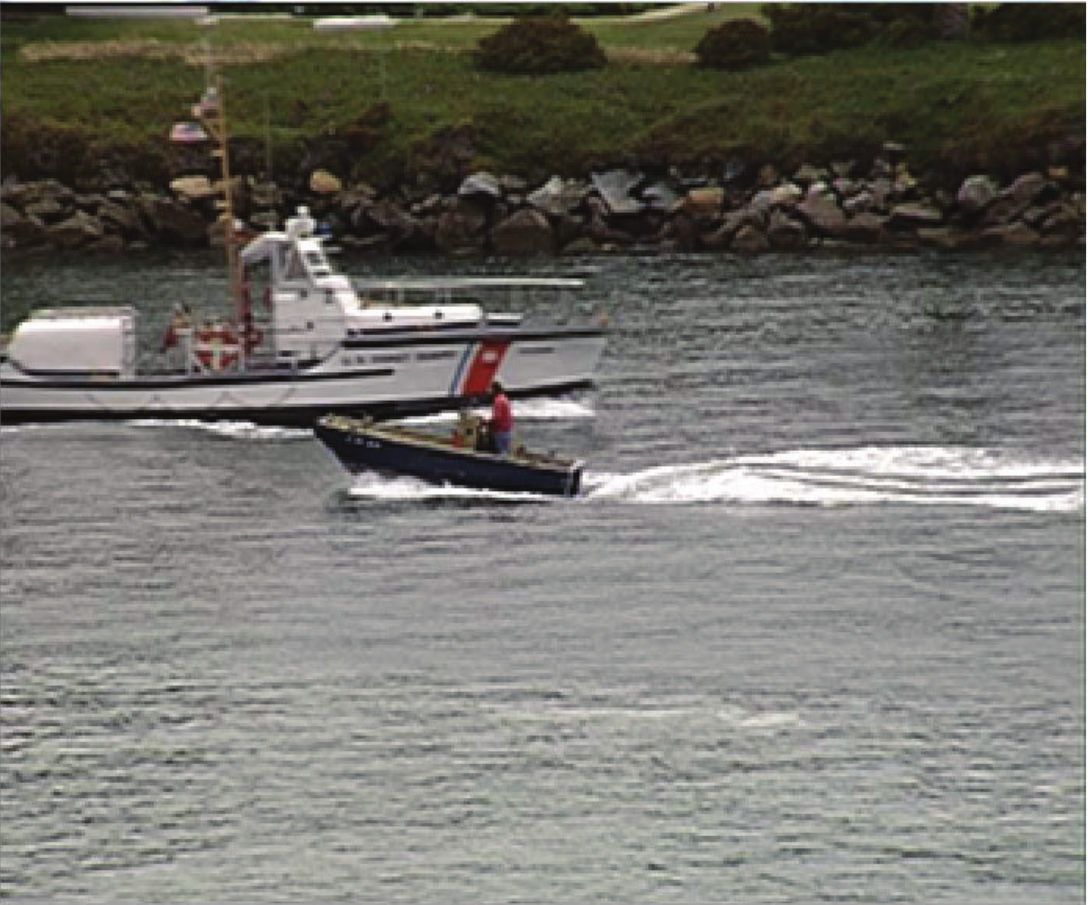}}
    \hspace{1pt}
    \subfigure[LyaRI: Frame index = 29]{
    \label{fig:subfig:d} 
    \includegraphics[width= 2.24 in, height = 1.3 in]{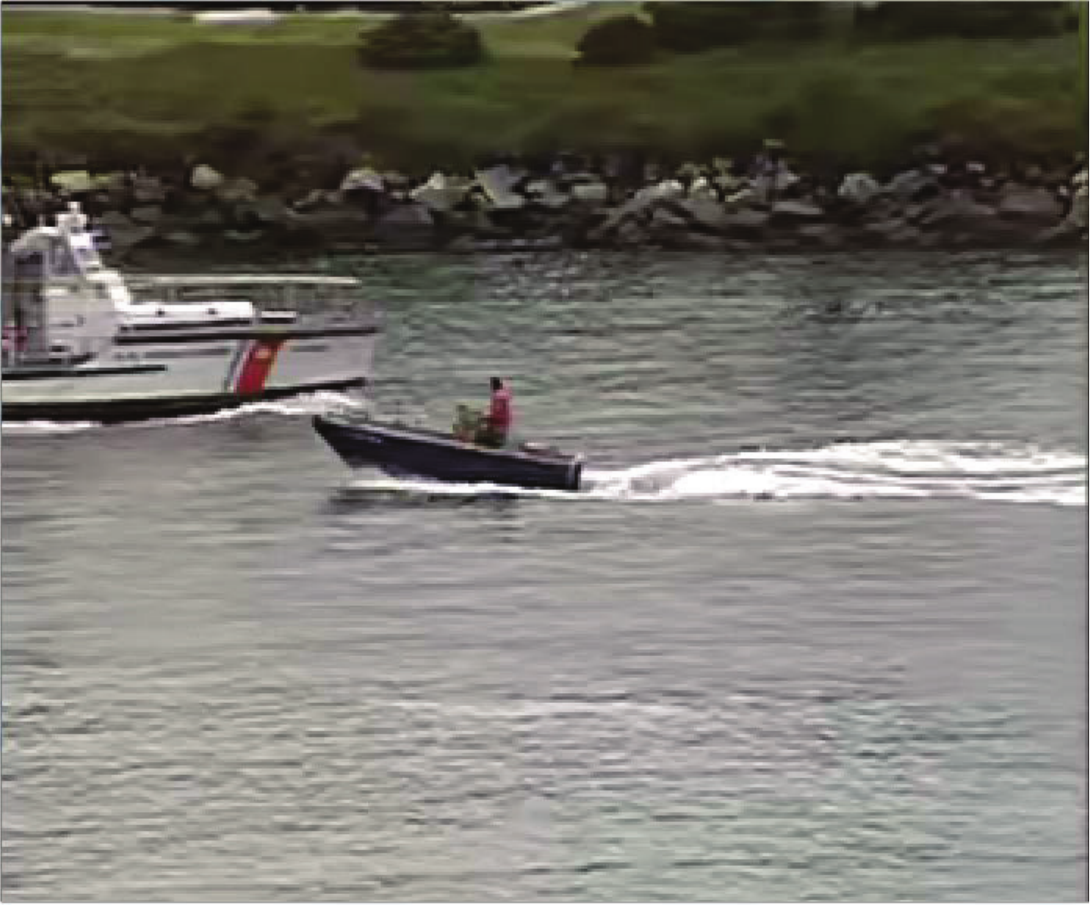}}
    \hspace{1pt}
    \subfigure[LyaRI: Frame index = 30]{
    \label{fig:subfig:e} 
    \includegraphics[width= 2.24 in, height = 1.3 in]{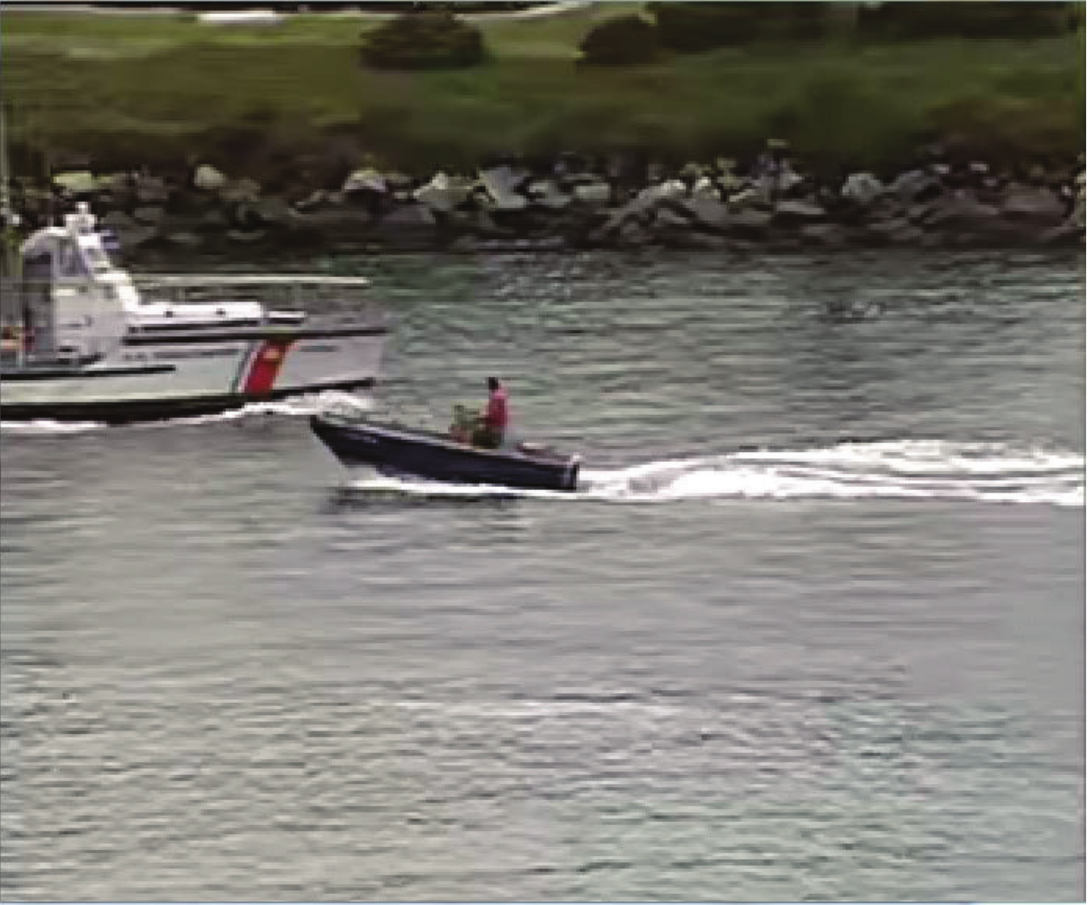}}
    \hspace{1pt}
    \subfigure[LyaRI: Frame index = 48]{
    \label{fig:subfig:f} 
    \includegraphics[width= 2.24 in, height = 1.3 in]{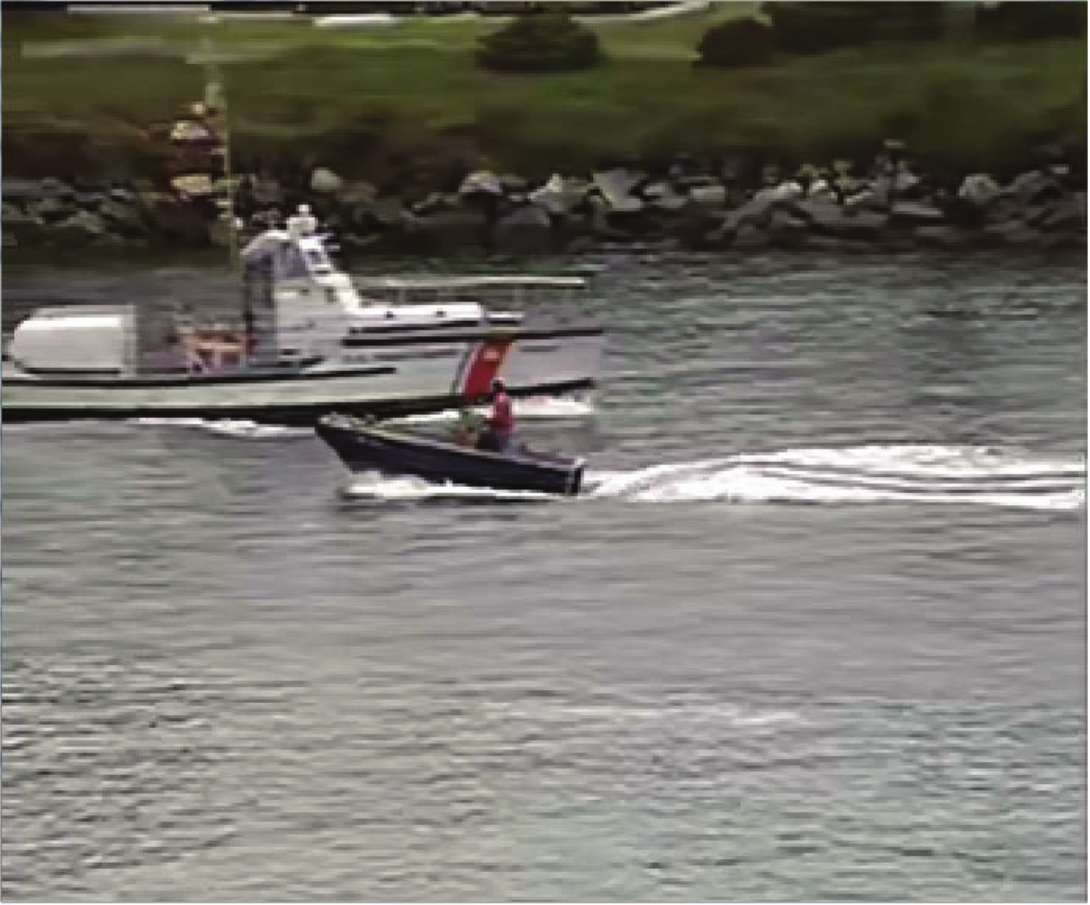}}
     \hspace{1pt}
    \subfigure[HM RC: Frame index = 29]{
    \label{fig:subfig:g} 
    \includegraphics[width= 2.24 in, height = 1.3 in]{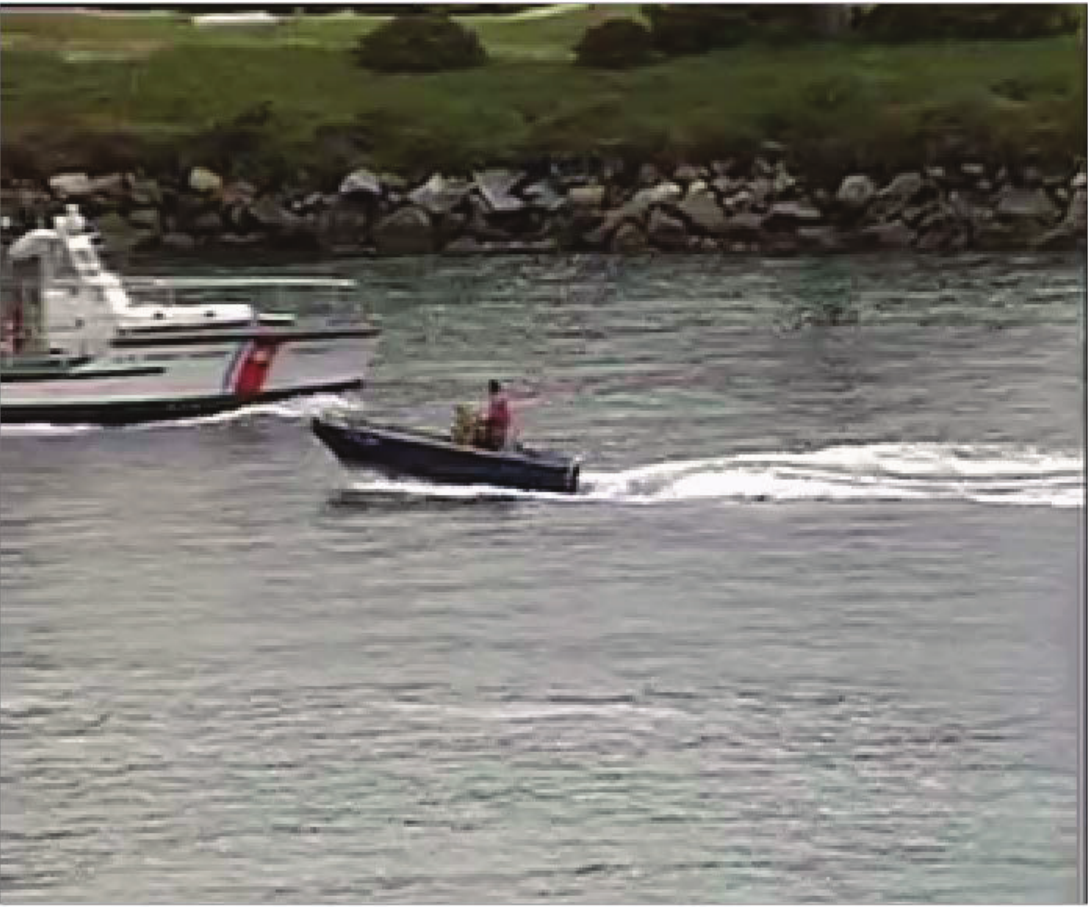}}
    \hspace{1pt}
    \subfigure[HM RC: Frame index = 30]{
    \label{fig:subfig:h} 
    \includegraphics[width= 2.24 in, height = 1.3 in]{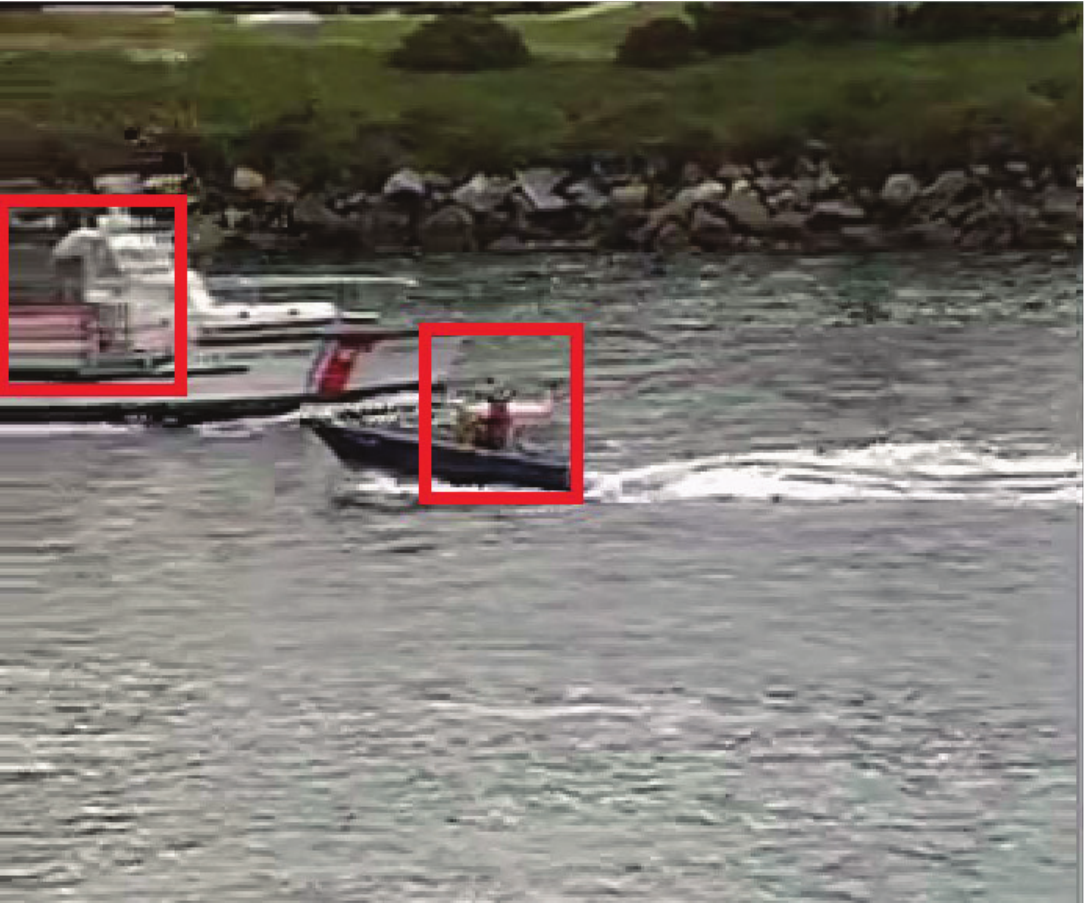}}
    \hspace{1pt}
    \subfigure[HM RC: Frame index = 48]{
    \label{fig:subfig:i} 
    \includegraphics[width= 2.24 in, height = 1.3 in]{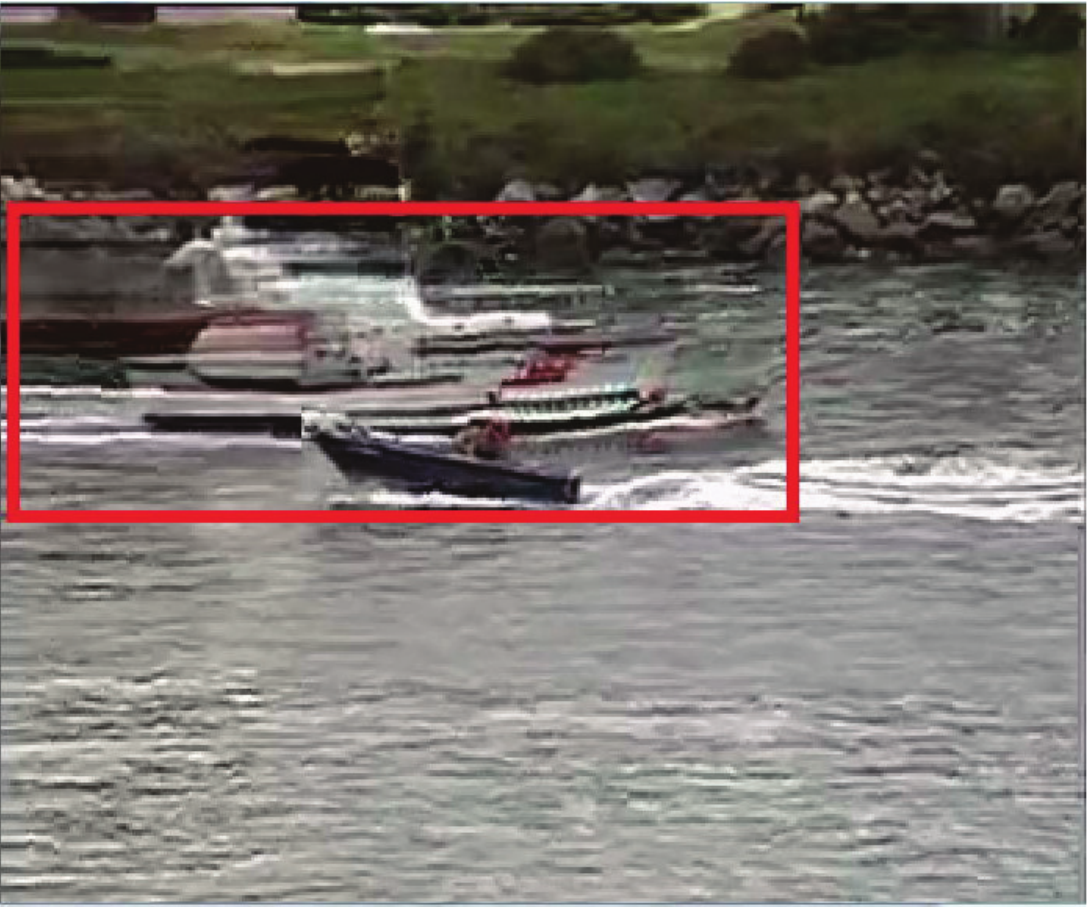}}
    \hspace{1pt}    
    \subfigure[Original: Frame index = 121]{
    \label{fig:subfig:j} 
    \includegraphics[width= 2.24 in, height = 1.3 in]{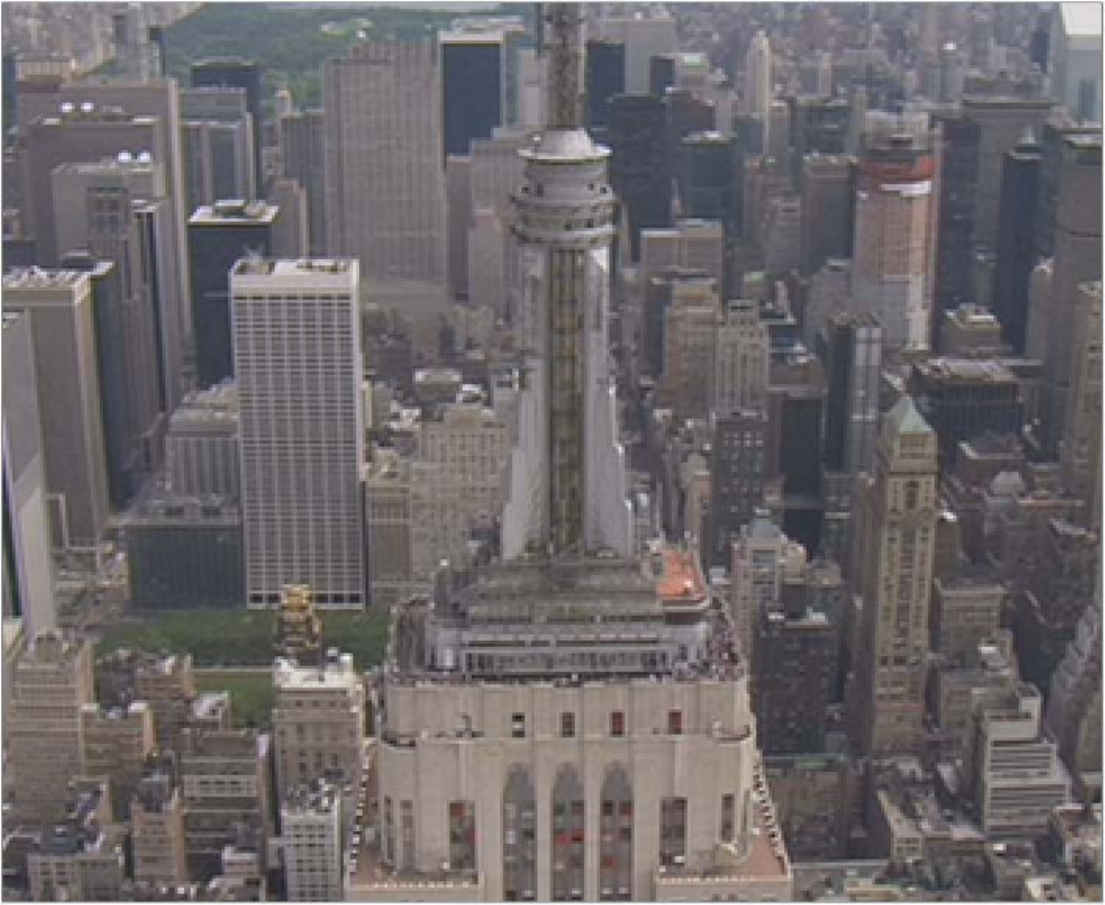}}
  \hspace{1pt}
    \subfigure[Original: Frame index = 170]{
    \label{fig:subfig:k} 
    \includegraphics[width= 2.24 in, height = 1.3 in]{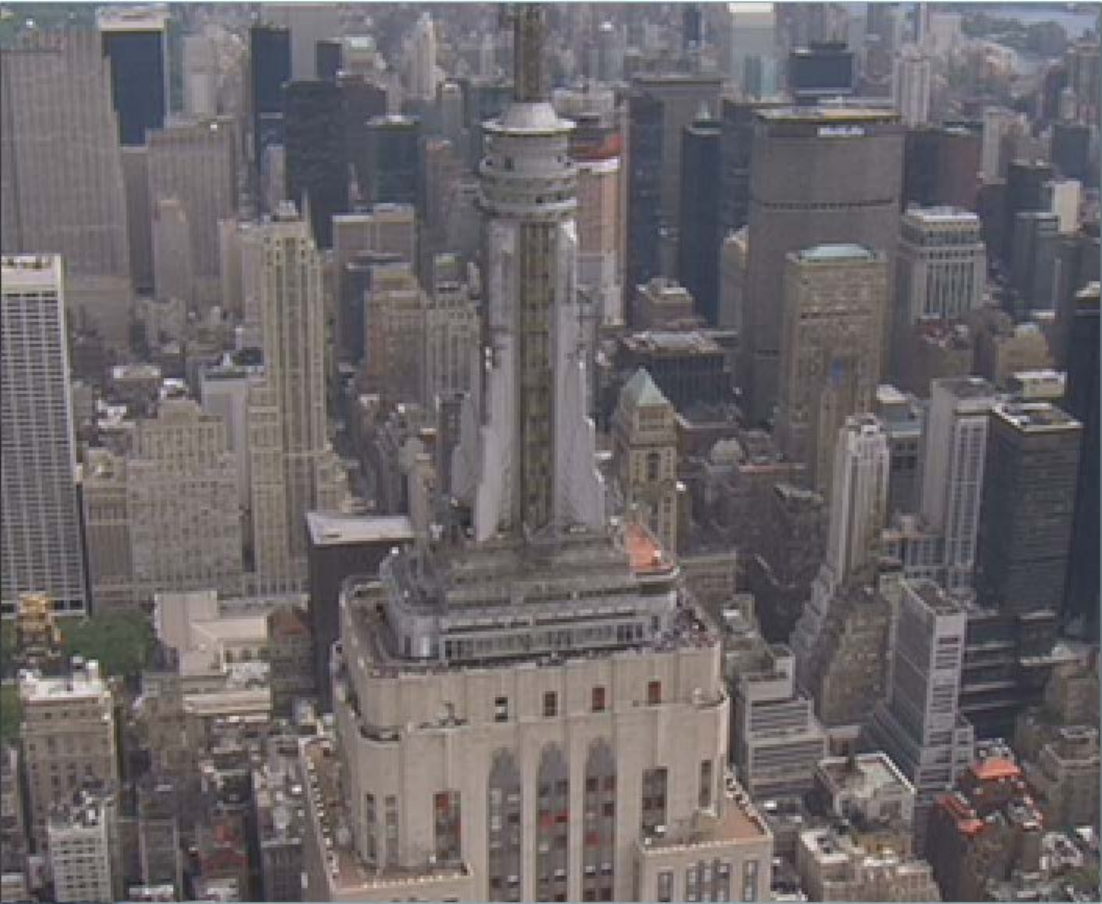}}
    \hspace{1pt}
    \subfigure[Original: Frame index = 200]{
    \label{fig:subfig:l} 
    \includegraphics[width= 2.24 in, height = 1.3 in]{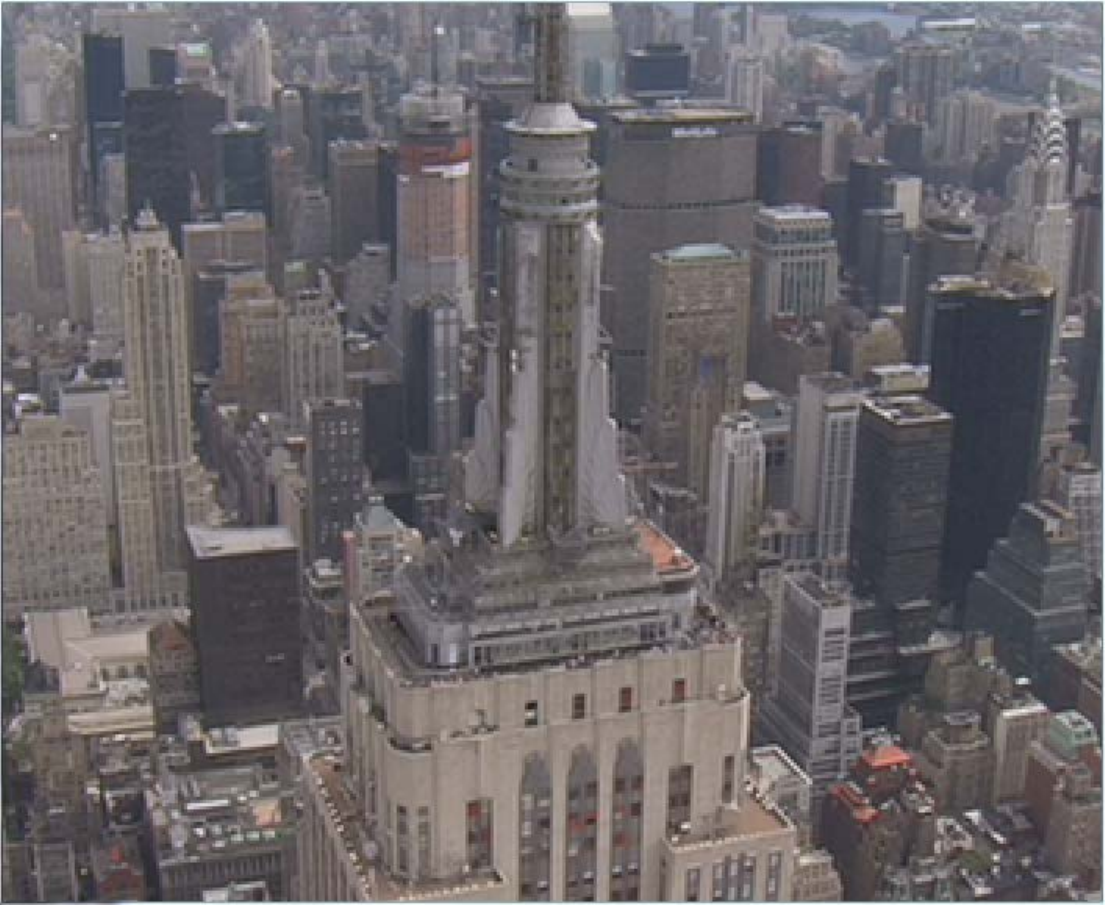}}
    \hspace{1pt}
    \subfigure[LyaRI: Frame index = 121]{
    \label{fig:subfig:m} 
    \includegraphics[width= 2.24 in, height = 1.3 in]{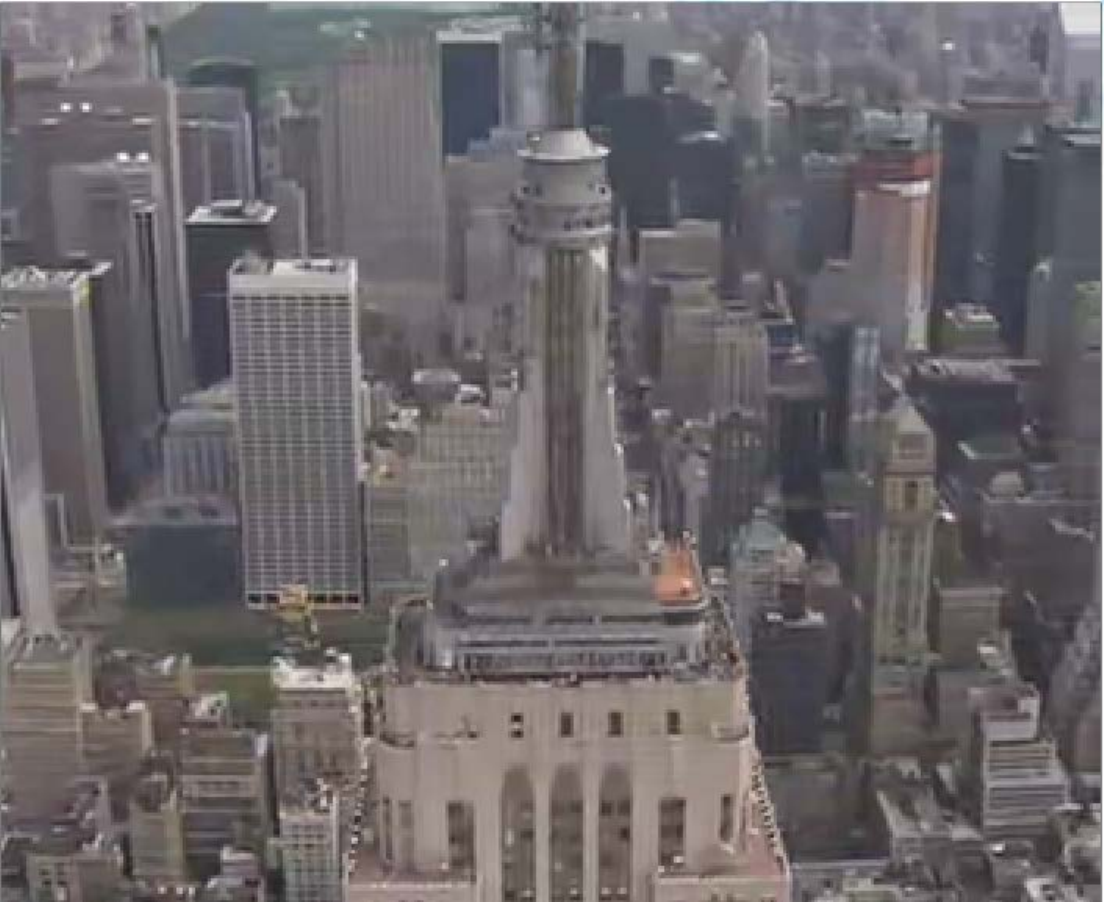}}
    \hspace{1pt}
    \subfigure[LyaRI: Frame index = 170]{
    \label{fig:subfig:n} 
    \includegraphics[width= 2.24 in, height = 1.3 in]{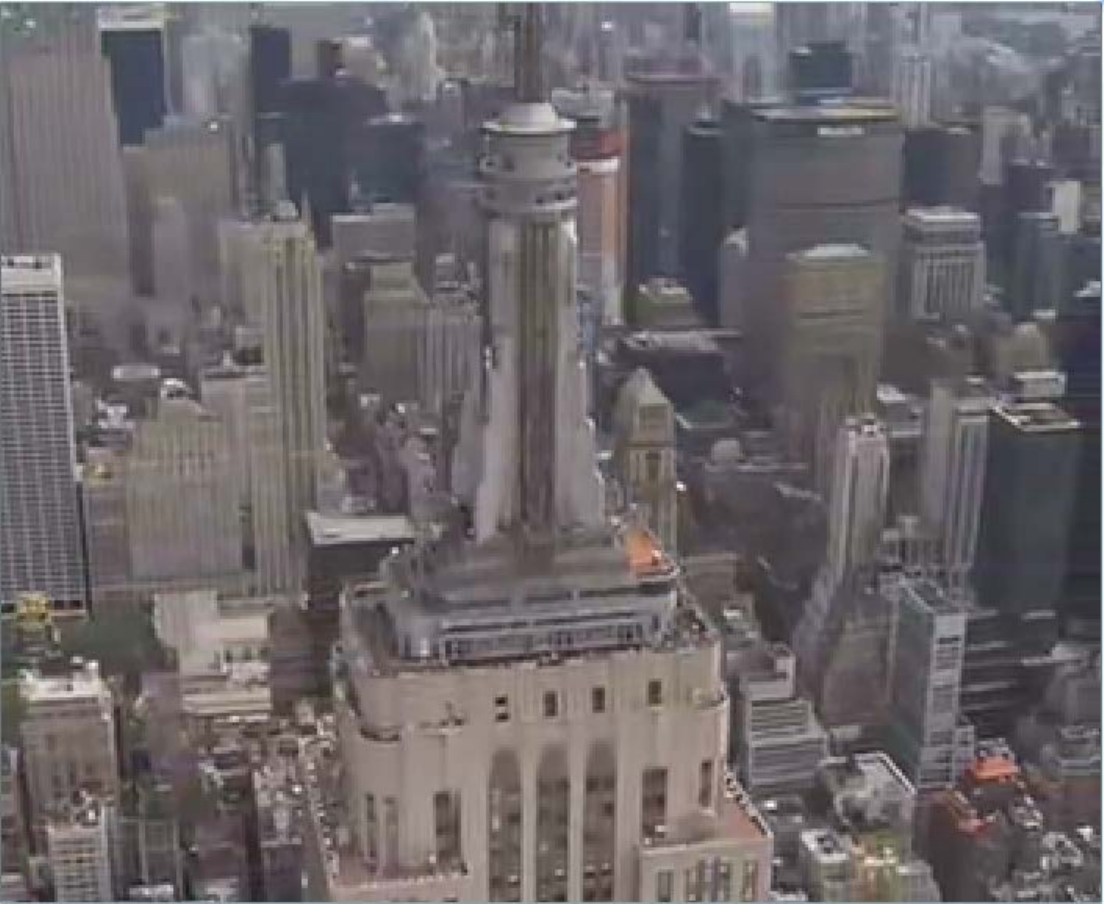}}
    \hspace{1pt}
    \subfigure[LyaRI: Frame index = 200]{
    \label{fig:subfig:o} 
    \includegraphics[width= 2.24 in, height = 1.3 in]{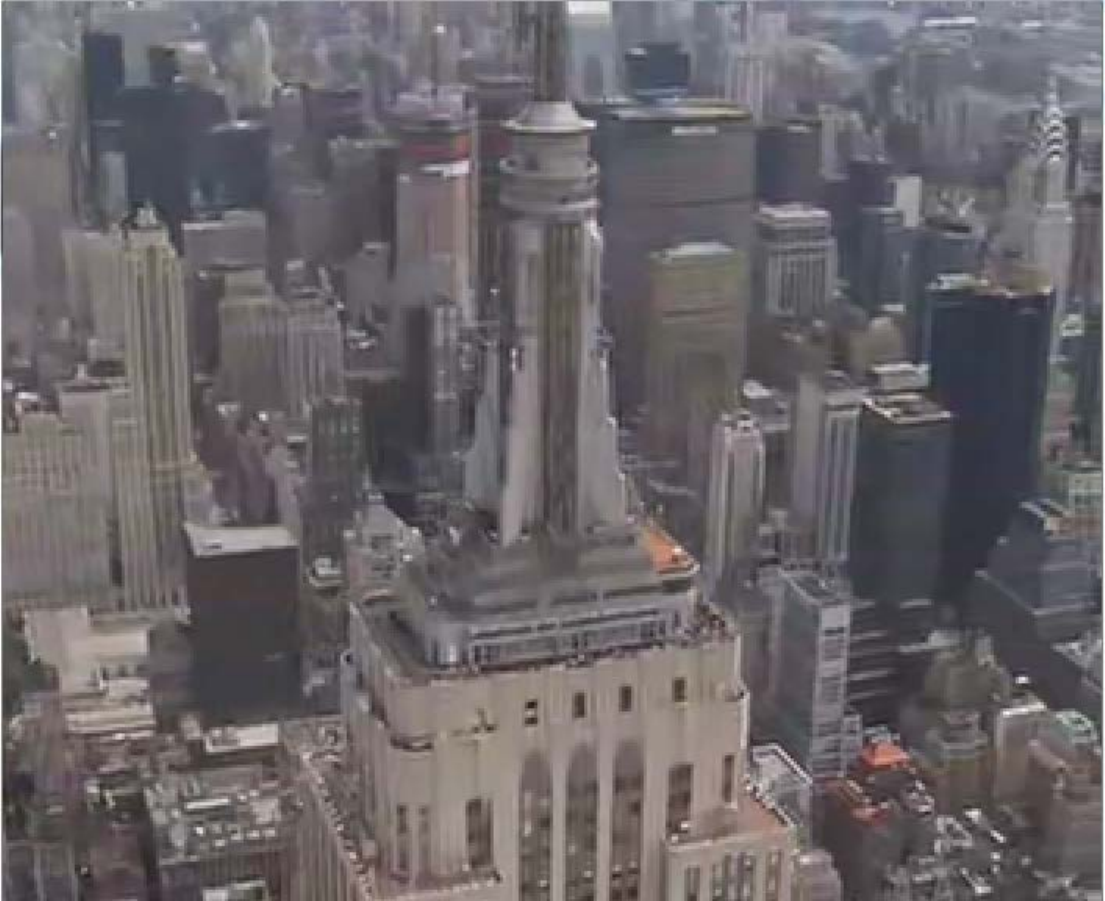}}
     \hspace{1pt}
    \subfigure[HM RC: Frame index = 121]{
    \label{fig:subfig:p} 
    \includegraphics[width= 2.24 in, height = 1.3 in]{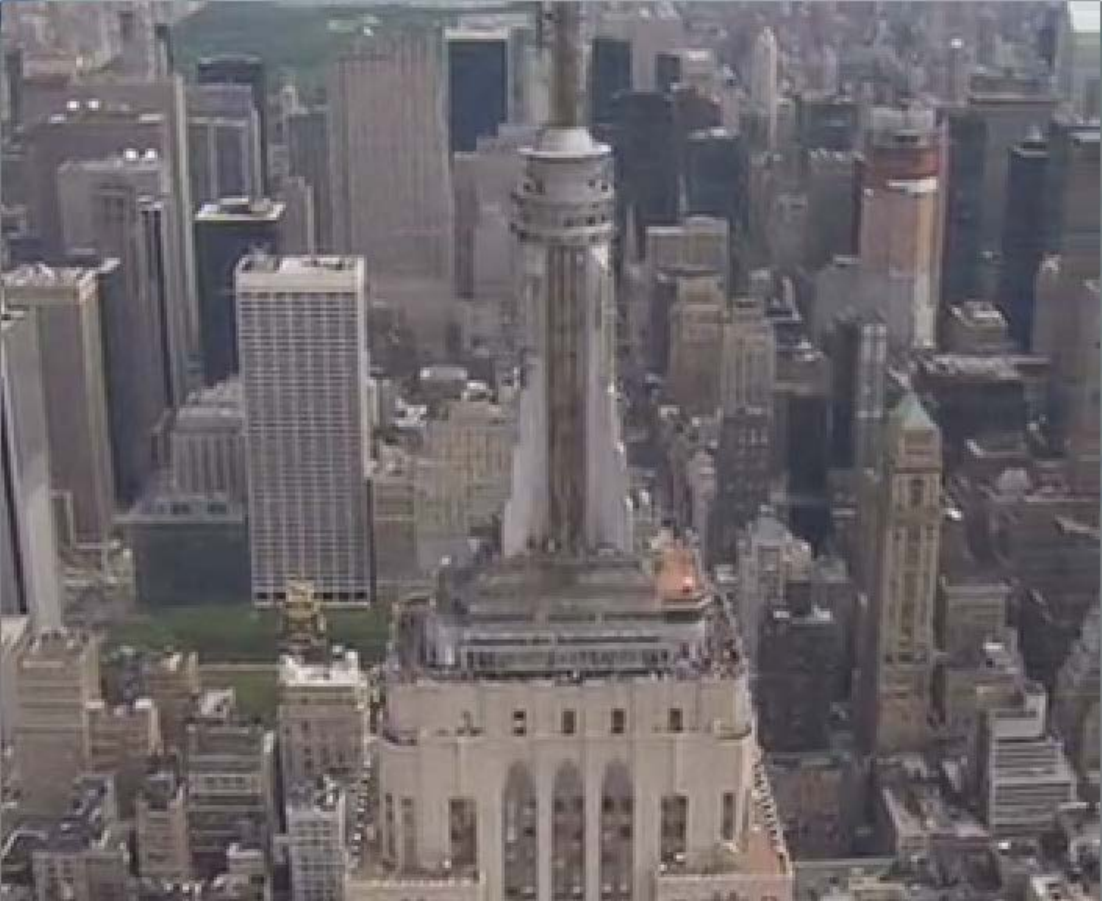}}
    \hspace{1pt}
    \subfigure[HM RC: Frame index = 170]{
    \label{fig:subfig:q} 
    \includegraphics[width= 2.24 in, height = 1.3 in]{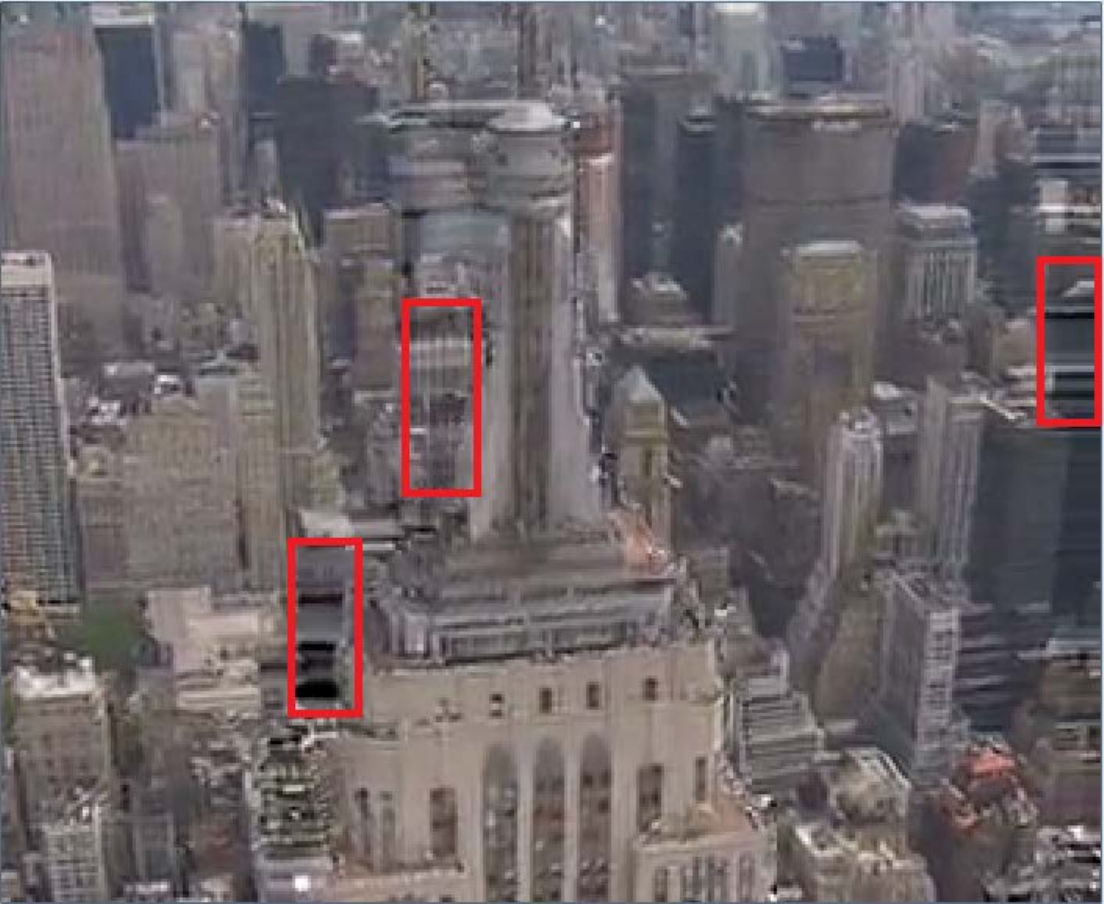}}
    \hspace{1pt}
    \subfigure[HM RC: Frame index = 200]{
    \label{fig:subfig:r} 
    \includegraphics[width= 2.24 in, height = 1.3 in]{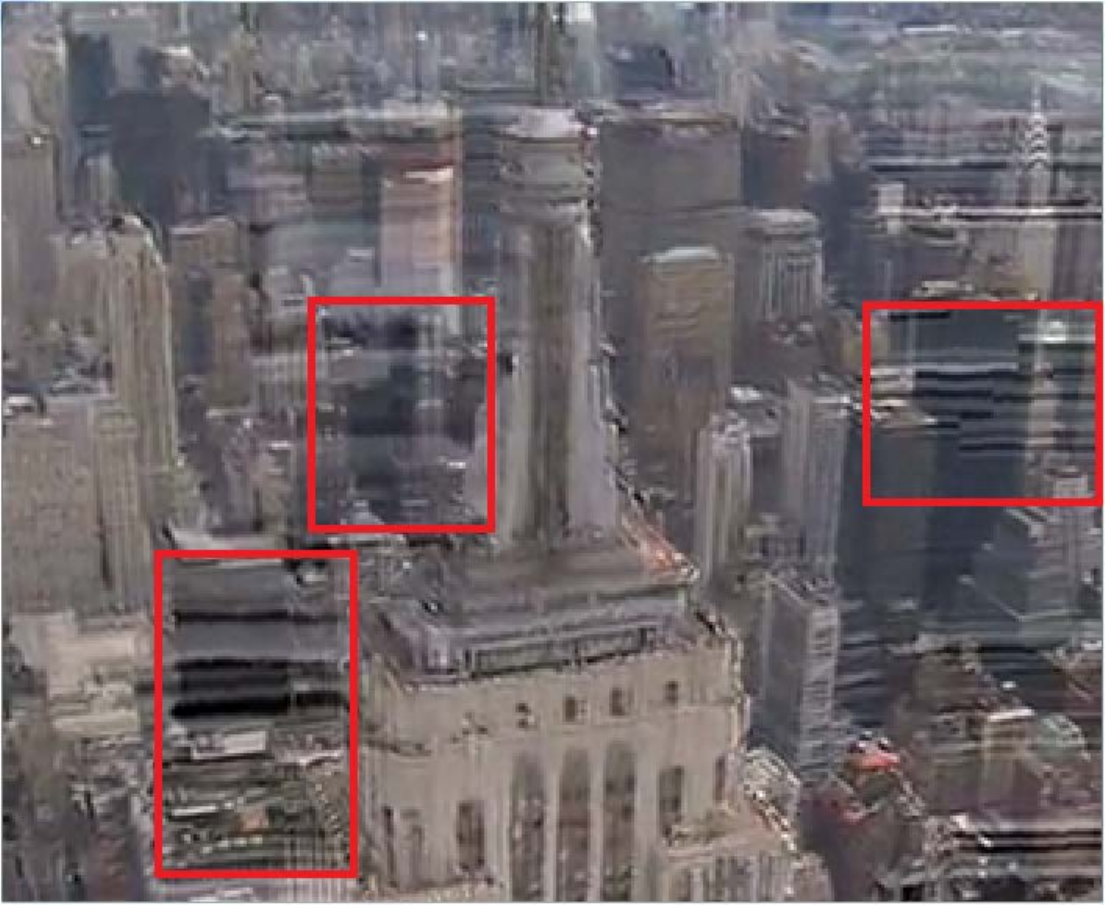}}
    \caption{{Subjective results for the Coastguard sequence and the City sequence.}}
\label{fig_subjective_LyaRI_HMRC}
\end{figure*}
\subsection{{Comparative Analysis of Subjective Performance}}
{In this subsection, we also design a subjective experiment to validate the effectiveness of LyaRI in practical video coding and transmission applications. The experiment is based on the Coastguard sequence and the City sequence. Both LyaRI and the comparative method HM RC adopt the same initial coding parameters. 
Experimental results indicate that there is no visual
artifacts or stalling in video decoding at the receiver using LyaRI. 
This is because LyaRI can adjust coding parameters to adapt to dynamic changes in the AtG channel. 
In contrast, HM RC method experiences visual artifacts at both the $30$-th frame of Coastguard and the $170$-th frame of City due to packet loss. 
Further, the packet loss results in a gradual degradation quality of subsequent video frames due to error propagation in video coding. For instance, Fig. \ref{fig_subjective_LyaRI_HMRC}(i) depicts the subjective quality of the $48$-th frame of Coastguard, and Fig. \ref{fig_subjective_LyaRI_HMRC}(r) the $200$-th frame of City.
It stems from the lack of a mechanism that adapts the source coding bitrate to the dynamic channel capacity, thus leading to packet loss.
Besides, from Fig. \ref{fig_subjective_LyaRI_HMRC}, it can be observed that the subjective quality of videos encoded by both LyaRI and HM RC has declined compared to the original YUV video. 
The inefficiency of the HM coding process affects the video quality.}

\section{Conclusion}
This paper investigated the joint source-channel optimization for UAV video coding and transmission. Building upon the research and construction of video coding d-P-R-D model and UAV channel transmission d-P-R-D model, this paper formulated a joint source-channel video coding and transmission optimization problem. The goal was to minimize the end-to-end distortion of UAV video coding and transmission and total power consumption, while meeting requirements of source-channel rate adaptability, end-to-end delay, and power consumption. A Lyapunov repeated iteration algorithm was proposed to solve this problem. {Both objective and subjective} experimental results verified the effectiveness of the constructed models and the proposed algorithm. 
{The results showed that, compared with the benchmark, the proposed algorithm achieved better video quality and stability performance, and the variance of its obtained encoding bitrate is reduced by 47.74\%.}
This paper {implemented} UAV video coding and transmission using nonliregression analysis and conventional optimization approaches. In the near future, how to explore generative or interactive artificial intelligence approaches for efficient and joint UAV video coding and transmission deserves to be studied in depth.


%

\bibliographystyle{IEEEtran}
\bibliography{Globecom_RAN_slicing}
\appendices
\section{Proof of Lemma \ref{lemma:lemma_video_coding_bitrate}}
The transformed residuals follow a zero-mean i.i.d Laplacian distribution. Its probability density function (PDF) can be expressed as
\begin{equation}\label{eq:Laplacian_PDF}
p(x) = \frac{\Delta }{2}{e^{ - \Delta \left| x \right|}}
\end{equation}
where $x$ represents the value of transformed residuals, and $\Delta  = {{\sqrt 2 } \mathord{\left/
 {\vphantom {{\sqrt 2 } \sigma }} \right.
 \kern-\nulldelimiterspace} \sigma }$ is the Laplacian parameter corresponding to the standard deviation of transformed residuals.

 According to the definition of source entropy, the video coding bitrate can be approximated by the entropy of quantized transformed residuals. The derivation is as follows
\begin{equation}\label{eq:entropy_transform_residual}
{R_e}(\Delta ,Q) \approx H(\Delta ,Q) =  - {P_0}{\log _2}{P_0} - 2\sum\limits_{n = 1}^\infty  {{P_n}{{\log }_2}{P_n}} 
\end{equation}

The probability of quantized transformed residuals being zero can be expressed as 
\begin{equation}\label{eq:probability_residual_zero}
{P_0} = \int_{ - (Q - \mu Q)}^{Q - \mu Q} {p(x)dx}  = 1 - {e^{\Delta Q(\mu  - 1)}}
\end{equation}
where for quantization step $Q$, $\mu Q$ denotes the rounding offset, and $\mu$ is a parameter between $(0,1)$.

The probability of transformed residuals falling within the n-th quantization interval is as follows
\begin{equation}\label{eq:probability_residual_n_interval}
{P_n} = \int_{nQ - \mu Q}^{(n + 1)Q - \mu Q} {p(x)dx}  = \frac{1}{2}(1 - {e^{ - \Delta Q}}){e^{\Delta Q(u - n)}}
\end{equation}

By substituting (\ref{eq:probability_residual_zero}) and (\ref{eq:probability_residual_n_interval}) into (\ref{eq:entropy_transform_residual}), the following equation can be obtained
\begin{equation}\label{eq:entropy_residual_substituting}
\begin{array}{l}
  {R_e}(\Delta ,Q) = 1 - {P_0}{\log _2}2{P_0} + (1 - {P_0}) \times  \hfill \\
  \left[ {\frac{{\Delta Q{{\log }_2}e}}{{1 - {e^{ - \Delta Q}}}} - {{\log }_2}\left[ {{e^{\mu \Delta Q}} - {e^{(\mu  - 1)\Delta Q}}} \right]} \right] \hfill \\ 
\end{array}
\end{equation}

Next, by substituting (\ref{eq:sigma_model}) into (\ref{eq:entropy_residual_substituting}), (\ref{eq:video_coding_bitrate}) is obtained, thereby completing the proof.

\section{Proof of Lemma \ref{lemma:lemma_upper_bound}}
According to (\ref{eq:virtual_queues}) and (\ref{eq:stability_conditions}), we discuss the upper bound of $\frac{1}{2}{\left( {{{[X(t + 1)]}^ + }} \right)^2}$ in three cases.

Case 1: when $X(t + 1) \geqslant 0$ and $X(t) \geqslant 0$, we have
\begin{equation}\label{eq:upper_bound_case1}
\begin{array}{l}
  \frac{1}{2}{\left( {{{[X(t + 1)]}^ + }} \right)^2} = {[X(t)]^ + }\left( {{R_e}(\lambda ,Q;t) - {R_c}(t)} \right) \hfill \\
   + \frac{1}{2}{\left( {{{[X(t)]}^ + }} \right)^2} + \frac{1}{2}{\left( {{R_e}(\lambda ,Q;t) - {R_c}(t)} \right)^2} \hfill \\ 
\end{array}
\end{equation}

Case 2: when $X(t + 1) \geqslant 0$ and $X(t) < 0$, it can be known that $0 \leqslant X(t + 1) < {R_e}(\lambda ,Q;t) - {R_c}(t),\forall t$. Further, we can obtain $\frac{1}{2}{\left( {{{[X(t + 1)]}^ + }} \right)^2} < \frac{1}{2}{\left( {{R_e}(\lambda ,Q;t) - {R_c}(t)} \right)^2}$. Considering that ${[X(t)]^ + } = 0$, we have
\begin{equation}\label{eq:upper_bound_case2}
\begin{array}{l}
  \frac{1}{2}{\left( {{{[X(t + 1)]}^ + }} \right)^2} < {[X(t)]^ + }\left( {{R_e}(\lambda ,Q;t) - {R_c}(t)} \right) +  \hfill \\
  \frac{1}{2}{\left( {{{[X(t)]}^ + }} \right)^2} + \frac{1}{2}{\left( {{R_e}(\lambda ,Q;t) - {R_c}(t)} \right)^2} \hfill \\ 
\end{array}
\end{equation}

Case 3: when $X(t + 1) < 0$, then $\frac{1}{2}{\left( {{{[X(t + 1)]}^ + }} \right)^2} = 0$, it can be inferred that
\begin{equation}\label{eq:upper_bound_case3}
\begin{array}{l}
  \frac{1}{2}{\left( {{{[X(t + 1)]}^ + }} \right)^2} \leqslant {[X(t)]^ + }\left( {{R_e}(\lambda ,Q;t) - {R_c}(t)} \right) +  \hfill \\
  \frac{1}{2}{\left( {{{[X(t)]}^ + }} \right)^2} + \frac{1}{2}{\left( {{R_e}(\lambda ,Q;t) - {R_c}(t)} \right)^2} \hfill \\ 
\end{array}
\end{equation}

In summary, based on the definition of $\Delta \left( t \right)$, we can obtain $\Delta \left( t \right) \leqslant {[X(t)]^ + }\left( {{R_e}(\lambda ,Q;t) - {R_c}(t)} \right) + \frac{1}{2}{\left( {{R_e}(\lambda ,Q;t) - {R_c}(t)} \right)^2}$. By adding $V({D_e}({R_e}(\lambda ,Q;t)) + {\rho _1}{D_c}({P_t}(t)) + {\rho _2}{P_{tot}}(t))$ to both sides of the inequality, we can obtain (\ref{eq:upper_bound_Lyapunov})
, thus proving the lemma.

\section{Proof of Lemma \ref{lemma:lemma_transform_convex}}
For (\ref{eq:reformed_Q_subproblem}c), a slack variable $\varepsilon $ is introduced such that ${R_e}(\lambda ,Q;t) \geqslant \varepsilon  \geqslant {(\frac{C}{\delta })^{\frac{1}{K}}}$. It is not difficult to conclude that ${R_e}(\lambda ,Q;t) \geqslant \varepsilon $ is a non-convex constraint. To effectively address this issue, at any given local iteration point ${Q_0}$, the SCA strategy can be employed to obtain its approximate convex constraint, represented as follows
\begin{equation}\label{eq:video_bitrate_SCA_Q0}
{R_e}({\lambda ^{(r)}},{Q_0};t) + \frac{{\partial {R_e}({\lambda ^{(r)}},{Q_0};t)}}{{\partial Q}}(Q - {Q_0}) \geqslant \varepsilon 
\end{equation}
For the non-convex constraint $\varepsilon  \geqslant {(\frac{C}{\delta })^{\frac{1}{K}}}$, by introducing a slack variable $\xi $, we have
\begin{equation}\label{eq:non-convex_transform_xi}
\ln \varepsilon  \geqslant \xi  \geqslant \frac{1}{K}(\ln C - \ln \delta ) 
\end{equation}

For the first inequality in (\ref{eq:non-convex_transform_xi}), we can derive the following constraint $\varepsilon  \geqslant {e^\xi }$. According to the standard form of exponential cone ${K_{\exp }} = \left\{ {x \in {R^3}:{x_1} \geqslant {x_2}\exp ({x_3}/{x_2}),{x_1},{x_2} \geqslant 0} \right\}$, the constraint $\varepsilon  \geqslant {e^\xi }$ can be transformed into an exponential cone, i.e.
\begin{equation}\label{eq:non-convex_transform_xi_left}
\left( {\varepsilon ,1,\xi } \right) \in {K_{\exp }}
\end{equation}

Similarly, for the constraint $\xi  \geqslant \frac{1}{K}(\ln C - \ln \delta )$, it can be converted into the following exponential cone
\begin{equation}\label{eq:non-convex_transform_xi_right}
\left( {\delta ,C, - CK\xi } \right) \in {K_{\exp }}
\end{equation}

This completes the proof.

\section{Proof of Lemma \ref{lemma:lemma_transform_data_rate}}
For the constraint ${R_c}(t) \geqslant \varphi $, since ${R_c}(t) = {\log _2}(1 + \frac{{{P_t}(t)}}{{{L_{AtG}}(t){P_n}}})$, we have ${\log _2}(1 + \frac{{{P_t}(t)}}{{{L_{AtG}}(t){P_n}}}) \geqslant \varphi $. Then, set ${{\rm Z}_1} = 1 + \frac{{{P_t}(t)}}{{{L_{AtG}}(t){P_n}}}$, the constraint can be converted into the following exponential cone
\begin{equation}\label{eq:data_rate_exponential_cone}
\left( {{{\rm Z}_1},1,\varphi \ln 2} \right) \in {K_{\exp }}
\end{equation}

For (\ref{eq:original_problem}c) and (\ref{eq:reformed_lambda_Pt_subproblem}f), a slack variable $\tau $ is introduced and set $\frac{L}{{B\varphi }} \leqslant \tau $. Then, these two constraints can be transformed into the following form
\begin{subequations}\label{eq:lambda_Pt_subproblem_constraints_transform}
\begin{alignat}{2}
  &\tau  \leqslant {d_{\max \_trans}} \hfill \\
  &{(2\lambda  + 1)^2}{d_{coe}} + \tau  \leqslant {d_{\max }} 
\end{alignat}
\end{subequations}

For the inequality $\frac{L}{{B\varphi }} \leqslant \tau $, based on the standard form of the rotated quadratic cone $Q_r^n = \left\{ {x \in {R^n}:2{x_1}{x_2} \geqslant \sum\limits_{j = 3}^n {x_j^2} ,{x_1} \geqslant 0,{x_2} \geqslant 0} \right\}$, it can be transformed into the following rotated quadratic cone
\begin{equation}\label{eq:first_rotated_quadratic_cone}
\left( {\varphi ,\tau ,\sqrt {\frac{{2L}}{B}} } \right) \in Q_r^3
\end{equation}

For the inequality ${(2\lambda  + 1)^2}{d_{coe}} + \tau  \leqslant {d_{\max }}$, defining ${{\rm Z}_2} = 2\lambda  + 1$ and ${{\rm Z}_3} = \frac{{{d_{\max }} - \tau }}{{{d_{coe}}}}$, the delay constraint can then be transformed into the following rotated quadratic cone
\begin{equation}\label{eq:second_rotated_quadratic_cone}
\left( {\frac{1}{2},{{\rm Z}_3},{{\rm Z}_2}} \right) \in Q_r^3
\end{equation}

This completes the proof.
\section{Proof of Lemma \ref{lemma:lemma_algorithm_convergence}}
Given a local point $({Q^{(r)}},{{P}_t}^{(r)}(t),{\lambda ^{(r)}})$ at the $r$-th iteration, and denote the corresponding value of (\ref{eq:Lyapunov_transformed_problem}) at this point as $\Psi ({{Q}^{(r)}},{{P}_t}^{(r)}(t),{\lambda ^{(r)}})$. By solving (\ref{eq:final_Q_subproblem}) at the $r+1$-th iteration we can obtain a solution ${{Q}^{(r + 1)}}$ such that $\Psi ({{Q}^{(r + 1)}},{{P}_t}^{(r)}(t),{\lambda ^{(r)}}) \leqslant \Psi ({{Q}^{(r)}},{{P}_t}^{(r)}(t),{\lambda ^{(r)}})$. 
Given the local point $({{Q}^{(r + 1)}},{{P}_t}^{(r)}(t),{\lambda ^{(r)}})$, we can obtain an updated solution ${{P}_t}^{(r + 1)}(t)$ and ${\lambda ^{(r + 1)}}$  by optimizing (\ref{eq:final_lambda_Pt_subproblem}) at the $r+1$-th iteration and have $\Psi ({{Q}^{(r + 1)}},{{P}_t}^{(r + 1)}(t),{\lambda ^{(r + 1)}}) \leqslant \Psi ({{Q}^{(r + 1)}},{{P}_t}^{(r)}(t),{\lambda ^{(r)}})$.
To this end, we can conclude that $\Psi ({{Q}^{(r + 1)}},{{P}_t}^{(r + 1)}(t),{\lambda ^{(r + 1)}}) \leqslant \Psi ({{Q}^{(r)}},{{P}_t}^{(r)}(t),{\lambda ^{(r)}})$. Besides, $\Psi ({{Q}^{(r)}},{{P}_t}^{(r)}(t),{\lambda ^{(r)}})$ is low-bounded at each iteration. Therefore, the iterative optimization Algorithm \ref{alg:alg1} is convergent.

Besides, Lemma \ref{lemma:lemma_upper_bound} points out that $\Delta \left( t \right) + V({D_e}({R_e}(\lambda ,Q;t)) + {\rho _1}{D_c}({P_t}(t)) + {\rho _2}{P_{tot}}(t))$ is upper-bounded at each time slot $t$. The time average of $X(t)$ tends to zero when $t \to \infty $. Therefore, Algorithm \ref{alg:alg1} can make the virtual queue mean-rate stable. 

This completes the proof.





\ifCLASSOPTIONcaptionsoff
  \newpage
\fi



%

\end{document}